\documentclass[a4paper,UKenglish,cleveref, autoref, thm-restate]{lipics-v2021}
\usepackage{mathtools}

\title{Coherent control and distinguishability of quantum channels via PBS-diagrams}

\author{Cyril Branciard}{Universit\'e Grenoble Alpes, CNRS, Grenoble INP, Institut N\'eel, F-38000 Grenoble, France}{cyril.branciard@neel.cnrs.fr}{https://orcid.org/0000-0001-9460-825X}{}

\author{Alexandre Cl\'ement}{Universit\'e de Lorraine, CNRS, Inria, LORIA, F-54000 Nancy, France}{alexandre.clement@loria.fr}{https://orcid.org/0000-0002-7958-5712}{}

\author{Mehdi Mhalla}{Universit\'e Grenoble Alpes, CNRS, Grenoble INP, LIG, F-38000 Grenoble, France}{mehdi.mhalla@univ-grenoble-alpes.fr}{https://orcid.org/0000-0003-4178-5396}{}

\author{Simon Perdrix}{Universit\'e de Lorraine, CNRS, Inria, LORIA, F-54000 Nancy, France}{simon.perdrix@loria.fr}{https://orcid.org/0000-0002-1808-2409}{}

\authorrunning{C. Branciard, A. Cl\'ement, M. Mhalla, and S. Perdrix} 
\Copyright{Cyril Branciard, Alexandre Cl\'ement, Mehdi Mhalla, and Simon Perdrix}



\category{} 

\relatedversion{} 

\supplement{}

\funding{This work is funded by ANR-17-CE25-0009 SoftQPro, ANR-17-CE24-0035 VanQuTe,
PIA-GDN/Quantex, LUE/UOQ, and by  \emph{``Investissements d'avenir''} (ANR-15-IDEX-02) program of
the French National Research Agency.}

\nolinenumbers 

\hideLIPIcs

\usepackage[svgnames,dvipsnames]{xcolor}
\usepackage{tikzit}

\tikzstyle{diamant}=[diamond, fill=couleurdefond, draw=black]
\tikzstyle{newe}=[rectangle, fill={gray!15}, draw=black, tikzit shape=rectangle, inner sep=0.2em]
\tikzstyle{cercle}=[circle, fill=couleurdefond, draw=black]
\tikzstyle{cartouche}=[rounded rectangle, fill=couleurdefond, draw=black]
\tikzstyle{neg}=[rounded rectangle, fill=couleurdefond, draw=black, execute at end node={$\neg$}]
\tikzstyle{sneg}=[rounded rectangle, fill=couleurdefond, draw=black, execute at end node={$\neg$}, scale=0.8]
\tikzstyle{negserie}=[rounded rectangle, fill=couleurdefond, draw=black, execute at end node={\footnotesize$\star\star$}]
\tikzstyle{diagrammevide}=[rectangle, fill=couleurdefond, draw=black, inner sep=1.25em, borddiagrammevide, tikzit shape=rectangle]
\tikzstyle{mdiagrammevide}=[rectangle, fill=couleurdefond, draw=black, inner sep=0.75em, sborddiagrammevide, tikzit shape=rectangle]
\tikzstyle{msdiagrammevide}=[rectangle, fill=couleurdefond, draw=black, inner sep=0.7em, msborddiagrammevide, tikzit shape=rectangle]
\tikzstyle{sdiagrammevide}=[rectangle, fill=couleurdefond, draw=black, inner sep=0.5em, sborddiagrammevide, tikzit shape=rectangle]
\tikzstyle{xsdiagrammevide}=[rectangle, fill=couleurdefond, draw=black, inner sep=0.4em, xsborddiagrammevide, tikzit shape=rectangle]
\tikzstyle{bs}=[shape=beam, fill=couleurdefond, draw, inner sep=0.25em, thick, tikzit fill=white]
\tikzstyle{sbs}=[shape=beam, fill=couleurdefond, draw, inner sep=0.2em, thick, tikzit fill=white]
\tikzstyle{boite22}=[fill=white, draw=black, shape=rectangle, minimum height=1cm, minimum width=0.5cm]
\tikzstyle{boite15}=[fill=white, draw=black, shape=rectangle, minimum height=0.7cm, minimum width=0.5cm]
\tikzstyle{boite2}=[fill=white, draw=black, shape=rectangle, minimum height=0cm, minimum width=0cm]
\tikzstyle{snegpotentiel}=[fill=couleurdefond, draw=black, shape=rounded rectangle, inner sep=0.25em, tikzit fill={rgb,255: red,191; green,191; blue,191}, execute at end node={\footnotesize$\star$}]
\tikzstyle{negpotentiel}=[fill=couleurdefond, draw=black, shape=rounded rectangle, tikzit fill={rgb,255: red,191; green,191; blue,191}, execute at end node={$\star$}]
\tikzstyle{token}=[fill=black, draw=black, shape=circle, inner sep=0.1em]
\tikzstyle{whitetoken}=[fill=white, draw=black, shape=circle, inner sep=0.1em]
\tikzstyle{boitePBS}=[fill=white, draw=gray, thick, shape=rectangle, rounded corners=3pt, minimum height=0.6cm, inner sep=0.1em, minimum width=0.5cm]
\tikzstyle{boitePBS2}=[fill=white, draw=gray, thick, shape=rectangle, rounded corners=3pt, minimum height=0.55cm, inner sep=0.1em, minimum width=0.5cm]

\tikzstyle{new}=[-]
\tikzstyle{tirets}=[-, draw=black, dashed]
\tikzstyle{noire}=[-, draw=black]
\tikzstyle{ep}=[-, draw=black]
\tikzstyle{longdashed}=[-, dash pattern=on 5pt off 5pt]
\tikzstyle{pointilles}=[-, draw=black, dotted]
\tikzstyle{grise}=[-, draw={rgb,255: red,191; green,191; blue,191}]
\tikzstyle{borddiagrammevide}=[-, dash pattern=on 0.5em off 0.5em on 0.5em off 0.5em on 0.5em off 0em]
\tikzstyle{msborddiagrammevide}=[-, dash pattern=on 0.28em off 0.28em on 0.28em off 0.28em on 0.28em off 0em]
\tikzstyle{sborddiagrammevide}=[-, dash pattern=on 0.2em off 0.2em on 0.2em off 0.2em on 0.2em off 0em]
\tikzstyle{xsborddiagrammevide}=[-, dash pattern=on 0.1em off 0.1em on 0.15em off 0.1em on 0.1em off 0em]

\input{figures/styles-pbs.tikzdefs}

\hypersetup{hypertexnames=false,raiselinks=true}
\usepackage{stmaryrd}
\usepackage{longtable}

\usepackage{empheq}

\newcolumntype{C}{>{$}c<{$}}  
\newcolumntype{R}{>{$}r<{$}}  
\newcolumntype{L}{>{$}l<{$}}  
\DeclareFontShape{OMX}{cmex}{m}{n}{
  <-7.5> cmex7
  <7.5-8.5> cmex8
  <8.5-9.5> cmex9
  <9.5-> cmex10
}{}
\SetSymbolFont{largesymbols}{normal}{OMX}{cmex}{m}{n}
\SetSymbolFont{largesymbols}{bold}  {OMX}{cmex}{m}{n}

\newcommand{\interp}[1]{\left\llbracket #1 \right\rrbracket}

\newenvironment{disarray}%
 {\everymath{\displaystyle\everymath{}}\array}%
 {\endarray}
\newenvironment{scriptarray}%
 {\everymath{\scriptstyle\everymath{}}\array}%
 {\endarray}
\renewcommand\H{\mathcal H} 
\newcommand\Hin{\mathcal H_{\textup{in}}} 
\newcommand\Hout{\mathcal H_{\textup{out}}}
\newcommand\Lin{\mathcal L} 
\newcommand\E{\mathcal E}
\renewcommand\S{\mathcal S}
\newcommand\N{\mathbb N}
\newcommand\C{\mathcal{C}}
\newcommand\CC{\mathbb C}
\newcommand\F{\mathcal F}
\newcommand\G{\mathcal G}
\renewcommand\L{\mathcal L}
\newcommand\hv{\{\rightarrow,\uparrow\}}

\makeatletter
\newlength{\negph@wd}
\DeclareRobustCommand{\negphantom}[1]{%
  \ifmmode
    \mathpalette\negph@math{#1}%
  \else
    \negph@do{#1}%
  \fi
}
\newcommand{\negph@math}[2]{\negph@do{$\m@th#1#2$}}
\newcommand{\negph@do}[1]{%
  \settowidth{\negph@wd}{#1}%
  \hspace*{-\negph@wd}%
}
\makeatother

\makeatletter
\newlength{\halfnegph@wd}
\DeclareRobustCommand{\halfnegphantom}[1]{%
  \ifmmode
    \mathpalette\halfnegph@math{#1}%
  \else
    \halfnegph@do{#1}%
  \fi
}
\newcommand{\halfnegph@math}[2]{\halfnegph@do{$\m@th#1#2$}}
\newcommand{\halfnegph@do}[1]{%
  \settowidth{\halfnegph@wd}{#1}%
  \hspace*{-0.5\halfnegph@wd}%
}
\makeatother

\makeatletter
\newlength{\halfph@wd}
\DeclareRobustCommand{\halfphantom}[1]{%
  \ifmmode
    \mathpalette\halfph@math{#1}%
  \else
    \halfph@do{#1}%
  \fi
}
\newcommand{\halfph@math}[2]{\halfph@do{$\m@th#1#2$}}
\newcommand{\halfph@do}[1]{%
  \settowidth{\halfph@wd}{#1}%
  \hspace*{0.5\halfph@wd}%
}
\makeatother

\newcommand\changelargeur[2]{#1\negphantom{#1}\phantom{#2}}
\newcommand\changelargeursicentre[2]{\halfnegphantom{#1}\halfphantom{#2}#1\halfnegphantom{#1}\halfphantom{#2}}

\newlength\traitsdiagrammevide
\newcommand\echellefils{0.35}

\newcommand{\urlalt}[2]{\href{#2}{\nolinkurl{#1}}}

\renewcommand{\ket}[1]{|{#1}\rangle}

\renewcommand{\bra}[1]{\langle{#1}|}

\newcommand{\typ}[3][\mathcal H]{{#1}^{(#3)}}

\begin{document}

\maketitle

\begin{abstract}
    
    We introduce a graphical language for coherent control of general quantum channels inspired by practical quantum optical setups involving polarising beam splitters (PBS). As  standard completely positive trace preserving maps  are known not to be appropriate to represent coherently controlled quantum channels, we propose to instead use purified channels, an extension of Stinespring's dilation. We characterise the observational equivalence of purified channels in various coherent-control contexts, paving the way towards a faithful representation of quantum channels under coherent control. 
    
\end{abstract}

\section{Introduction}

 Unlike the usual sequential and parallel compositions, coherent control allows one to perform two or more quantum evolutions in superposition. It is fairly easy with quantum optics---an important player in the development of quantum technologies---to construct setups that perform some coherent control. A polarising beam splitter (PBS) precisely allows one to do that: by reflecting for instance horizontally polarised particles and transmitting vertically polarised ones, it lets the polarisation control the path, and thereby the physical devices encountered, 
in a coherent way~\cite{friis14,rambo16}. This finds some interesting applications for quantum information processing (e.g., for error filtration~\cite{gisin05}), including the ability to perform some operations in an indefinite causal order, as for instance in the so-called quantum switch~\cite{chiribella13,araujo14,goswami18}, where the order in which two quantum operations are applied is controlled by the state of a qubit.

General quantum evolutions---a.k.a. quantum channels---are commonly represented as completely positive trace preserving (CPTP) maps. CPTP maps can naturally be composed in sequence and in parallel. However, it has been realised that the description of quantum channels in terms of CPTP maps is not appropriate for some particular setups involving coherent control~\cite{oi03,Abbott2020communication,chiribella19,kristjansson19}. One indeed needs some more information about their practical implementation to unambiguously determine the behaviour of such setups, and it was recently proposed to complete the description of channels by so-called transformation matrices~\cite{Abbott2020communication}, or vacuum extensions~\cite{chiribella19,kristjansson19}.

Here we consider a general class of setups involving PBS, and study how these can be used to coherently control quantum channels. We build upon the graphical language of PBS-diagrams introduced in~\cite{clement2020pbs}, in which the controlled operations were ``pure'' (typically, unitary), and extend it to allow for the control of more general quantum channels. As the description of channels as CPTP maps is inadequate here, we propose to work with \emph{purified channels} based on a unitary extension of Stinespring's dilation~\cite{Stinespring55Positive}. 

We address the question of the observational equivalence of purified channels, and show that different purified channels can be indistinguishable. To do so, we use PBS-diagrams to formalise three kinds of \emph{contexts}: when the context is PBS-free, we recover that two purified channels are indistinguishable if and only if they lead to the same CPTP map. When the context allows for PBS but no polarisation flips, we recover the characterisation in terms of superoperators and transformation matrices which was introduced for a particular setup~\cite{Abbott2020communication}. When we allow for arbitrary contexts, we obtain a characterisation of observational equivalence involving  ``second-level'' superoperators and transformation matrices. We finally open the discussion to more general coherent-control settings, and propose a refined equivalence relation as a candidate for characterising channel (in)distinguishabilty in such scenarios.

\section{PBS-diagrams}

\textup{PBS}-diagrams were introduced in~\cite{clement2020pbs} as a language for coherent control of ``pure'' quantum evolutions.
They aim at describing practical scenarios where a flying particle goes through an experimental setup, and is routed via polarising beam splitters. In addition to its polarisation, the particle carries some ``data'' register, whose state is described in some Hilbert space $\H$, and on which a number ``pure'' linear (typically, unitary) operators are applied.

Here  we shall enrich the pure PBS-diagram language so as to incorporate the coherent control of more general quantum channels. 
To this purpose, we start by defining an abstract version of \textup{PBS}-diagrams that we call \emph{bare diagrams}, and which we equip with a word path semantics describing the trajectory and change of polarisation of a particle that enters the diagram through some given input wire: the  word path semantics gives its new polarisation and position at the output of the diagram, together with  a word over some alphabet describing the sequence of \emph{bare gates}---where the quantum channels we want to control are located---crossed.
Subscribing to the idea that any general quantum operation can be seen as a unitary evolution of the system under consideration and its environment, we then define \emph{purified channels}, which can be coherently controlled in a similar way to the PBS-diagrams of~\cite{clement2020pbs}.
Replacing bare gates with purified channels, we obtain an extension%
\footnote{Strictly speaking, the PBS-diagrams of~\cite{clement2020pbs} did not require the operations inside the gates to be unitary, while here we impose such a restriction \emph{a priori}. One could however also consider non-unitary operations in our framework here, although one would lose our motivation based on the unitary extension of Stinespring's dilation.}
of the graphical language of~\cite{clement2020pbs}, which we call \emph{extended \textup{PBS}-diagrams} and which we equip with a quantum semantics obtained after discarding the (inaccessible) environments of all gates.

\subsection{Bare PBS-diagrams}
\subsubsection{Syntax}

A \emph{bare $\textup{PBS}$-diagram} is made of polarising beam splitters $\input{beamsplitter-xs.tikz}$, polarisation flips $\begin{tikzpicture}[scale=0.8]
	\begin{pgfonlayer}{nodelayer}
		\node [style=cartouche] (0) at (0, 0) {\tiny$\neg$};
		\node [style=none] (1) at (-0.8, 0) {};
		\node [style=none] (2) at (0.8, 0) {};
	\end{pgfonlayer}
	\begin{pgfonlayer}{edgelayer}
		\draw [style=noire] (1.center) to (0);
		\draw [style=noire] (0) to (2.center);
	\end{pgfonlayer}
\end{tikzpicture}
$, and bare gates $\begin{tikzpicture}[scale=0.8]
	\begin{pgfonlayer}{nodelayer}
		\node [style=newe] (0) at (0, 0) {\phantom{$a$}};
		\node [style=none] (1) at (-1.2, 0) {};
		\node [style=none] (2) at (1.2, 0) {};
		\node [style=none] (3) at (0.8, -0.35) {\footnotesize$a$};
	\end{pgfonlayer}
	\begin{pgfonlayer}{edgelayer}
		\draw [style=new] (1.center) to (2.center);
	\end{pgfonlayer}
\end{tikzpicture}
$.
Every bare gate is indexed by a unique label (here, $a$) used to identify the gate in the diagram.
These building blocks are connected via wires represented using the identity $\begin{tikzpicture}
	\begin{pgfonlayer}{nodelayer}
		\node [style=none] (0) at (-0.3, 0) {};
		\node [style=none] (1) at (0.3, 0) {};
	\end{pgfonlayer}
	\begin{pgfonlayer}{edgelayer}
		\draw [style=noire] (0.center) to (1.center);
	\end{pgfonlayer}
\end{tikzpicture}$ or the swap $\input{swap-xs.tikz}$.
The empty diagram is denoted by $\begin{tikzpicture}[scale=0.3]
	\begin{pgfonlayer}{nodelayer}
		\node [style=sdiagrammevide] (0) at (0, 0) {};
	\end{pgfonlayer}
\end{tikzpicture}
$. Diagrams can  be combined by means of sequential composition $ \circ $, parallel composition $\oplus$,\footnote{Denoted $\otimes$ in~\cite{clement2020pbs}. Here we change the notation to reflect how the parallel composition affects the structure of the Hilbert space describing the position of the particle (see Section~\ref{subsec:extendedPBS}).} and trace $Tr(\cdot)$, which represents a feedback loop.

We define a  typing judgement $\Gamma \vdash D : n$, where $\Gamma$ is the alphabet containing all gate indices,%
\footnote{We may write simply $D : n$, or even just $D$, when $\Gamma$ is not relevant or is clear from the context.}
to guarantee that the diagrams are well-formed---in particular, that the gate indices are unique---using a linear typing discipline:
\begin{definition}[Bare PBS-diagram]
 \label{def:barePBS-diagrams}
A  bare \textup{PBS}-diagram ${\Gamma} \vdash D:n$ (with $n \in{ \mathbb N}$) is inductively defined as: 
$$
\emptyset\vdash
\begin{tikzpicture}[scale=0.3]
	\begin{pgfonlayer}{nodelayer}
		\node [style=sdiagrammevide] (0) at (0, 0) {};
	\end{pgfonlayer}
\end{tikzpicture}
:0\qquad
\emptyset\vdash
\begin{tikzpicture}
	\begin{pgfonlayer}{nodelayer}
		\node [style=none] (0) at (-0.3, 0) {};
		\node [style=none] (1) at (0.3, 0) {};
	\end{pgfonlayer}
	\begin{pgfonlayer}{edgelayer}
		\draw [style=noire] (0.center) to (1.center);
	\end{pgfonlayer}
\end{tikzpicture}:1\qquad
\emptyset\vdash
\begin{tikzpicture}[scale=0.8]
	\begin{pgfonlayer}{nodelayer}
		\node [style=cartouche] (0) at (0, 0) {\tiny$\neg$};
		\node [style=none] (1) at (-0.8, 0) {};
		\node [style=none] (2) at (0.8, 0) {};
	\end{pgfonlayer}
	\begin{pgfonlayer}{edgelayer}
		\draw [style=noire] (1.center) to (0);
		\draw [style=noire] (0) to (2.center);
	\end{pgfonlayer}
\end{tikzpicture}
:1\qquad
\emptyset\vdash\input{swap-xs.tikz}:2\qquad
\emptyset\vdash
\input{beamsplitter-xs.tikz}:2\qquad
\{a\} \vdash \begin{tikzpicture}[scale=0.8]
	\begin{pgfonlayer}{nodelayer}
		\node [style=newe] (0) at (0, 0) {\phantom{$a$}};
		\node [style=none] (1) at (-1.2, 0) {};
		\node [style=none] (2) at (1.2, 0) {};
		\node [style=none] (3) at (0.8, -0.35) {\footnotesize$a$};
	\end{pgfonlayer}
	\begin{pgfonlayer}{edgelayer}
		\draw [style=new] (1.center) to (2.center);
	\end{pgfonlayer}
\end{tikzpicture}
:{1}
$$
$$
\dfrac{
\Gamma_1\!\vdash\!D_1\!:\!n \ \ \
\Gamma_2\!\vdash\!D_2\!:\!n \ \ \
\Gamma_1\!\cap\!\Gamma_2\!=\!\emptyset}{
{\Gamma _1} \cup {\Gamma _2} \vdash
D_2\circ D_1:{n}}
\quad
\dfrac{
{\Gamma _1} \!\vdash\!
D_1\!:\!{n_1}  \ \ \
{\Gamma _2} \!\vdash \!
D_2 \!:\!{n_2}
\ \ \ \Gamma_1\!\cap\! \Gamma_2\!=\!\emptyset}{
{\Gamma _1} \cup {\Gamma _2} \vdash 
D_1\oplus D_2:{n_1+n_2}}
\quad
\dfrac{
{\Gamma} \vdash 
D:{n+1}}{
{\Gamma} \vdash     
Tr(D):{n}}
$$
\end{definition}

\subparagraph*{Graphical representation.} PBS-diagrams  form a graphical language: compositions and trace are respectively depicted as follows {(for diagrams generically depicted as $\input{fig_generic_D.tikz}$)}: \\ \centerline{${\input{figCyril_generic_D2.tikz} \circ \input{figCyril_generic_D1.tikz} =} \input{composition-s.tikz}\qquad {\input{figCyril_generic_D1.tikz} \oplus \input{figCyril_generic_D2.tikz} =} \input{produittensoriel-s.tikz}\qquad {Tr\left(\input{figCyril_generic_Dn1.tikz}\right) = }\input{trace-s.tikz}$}

\medskip

Examples of bare PBS-diagrams are given in Fig.~\ref{fig:barePBS} below. Note that two \emph{a priori} distinct  constructions, like for instance $Tr(\begin{tikzpicture}[scale=0.8]
	\begin{pgfonlayer}{nodelayer}
		\node [style=newe] (0) at (0, 0) {\phantom{$a$}};
		\node [style=none] (1) at (-1.2, 0) {};
		\node [style=none] (2) at (1.2, 0) {};
		\node [style=none] (3) at (0.8, -0.35) {\footnotesize$a$};
	\end{pgfonlayer}
	\begin{pgfonlayer}{edgelayer}
		\draw [style=new] (1.center) to (2.center);
	\end{pgfonlayer}
\end{tikzpicture}
 \oplus \input{beamsplitter-xs.tikz})$ and $\begin{tikzpicture}[scale=0.8]
	\begin{pgfonlayer}{nodelayer}
		\node [style=newe] (0) at (0, 0) {\phantom{$a$}};
		\node [style=none] (1) at (-1.2, 0) {};
		\node [style=none] (2) at (1.2, 0) {};
		\node [style=none] (3) at (0.8, -0.35) {\footnotesize$a$};
	\end{pgfonlayer}
	\begin{pgfonlayer}{edgelayer}
		\draw [style=new] (1.center) to (2.center);
	\end{pgfonlayer}
\end{tikzpicture}
 \oplus Tr(\input{beamsplitter-xs.tikz})$, can lead to the same graphical representation $\input{bareasurtrbs_.tikz}$. To avoid ambiguity, we define diagrams modulo a structural congruence detailed in Appendix~\ref{appendix_congruence}. Roughly speaking, the structural congruence guarantees that (\emph{i}) two constructions leading to the same graphical representation are equivalent, and (\emph{ii}) a diagram can be deformed at will (without changing its topology), e.g.:
{\tikzset{tikzfig/.style={baseline=-0.25em,scale=\echellefils,every node/.style={scale=0.7}}}
\[
\!\!\!\!\input{swapswap-xs.tikz}=\input{filsparalleleslongs-xs.tikz}\qquad \input{natswap1.tikz}=\input{natswap2.tikz}\qquad \input{yankingvariantecentresurfil-xs.tikz}=\begin{tikzpicture}[scale=0.6]
	\begin{pgfonlayer}{nodelayer}
		\node [style=none] (0) at (-2, 0) {};
		\node [style=none] (1) at (2, 0) {};
	\end{pgfonlayer}
	\begin{pgfonlayer}{edgelayer}
		\draw [style=new] (0.center) to (1.center);
	\end{pgfonlayer}
\end{tikzpicture}\qquad \input{tracegrandfgb-xs.tikz}=\input{tracegbgrandf-xs.tikz}\]}

Note in particular that the length of the wires does not matter. Physically, if these diagrams were to be realised in practical setups, this would mean that the experiment should be insensible to the time at which the particle would go through the various elements; if needed one could always add (possibly polarisation-dependent) delay lines (e.g., $\input{delayline.tikz}$) to correct for a possible time mismatch between different paths.

\subsubsection{Word path semantics}

The word path semantics  describes the trajectory of a particle which enters a bare PBS-diagram $\Gamma\vdash D:n$ with a polarisation in the standard basis state $c\in \hv$ (horizontal or vertical) and from a definite position $p \in [n] \coloneqq \{0,\ldots,n-1\}$. Because of the polarising beam splitters, the trajectory of the particle depends on its polarisation: we take it to be reflected when the polarisation is horizontal, and transmitted when the polarisation is vertical. The ``negation'' $\begin{tikzpicture}[scale=0.8]
	\begin{pgfonlayer}{nodelayer}
		\node [style=cartouche] (0) at (0, 0) {\tiny$\neg$};
		\node [style=none] (1) at (-0.8, 0) {};
		\node [style=none] (2) at (0.8, 0) {};
	\end{pgfonlayer}
	\begin{pgfonlayer}{edgelayer}
		\draw [style=noire] (1.center) to (0);
		\draw [style=noire] (0) to (2.center);
	\end{pgfonlayer}
\end{tikzpicture}
$ flips the polarisation, while the gates do not act on the polarisation. The word path semantics of a diagram describes, given an initial polarisation and position, the final polarisation and position together with the sequence of gates---represented by a word over $\Gamma$---that the particle goes through:

\begin{definition}[Word path semantics]\label{wordpathsemantics} Given a bare $\textup{PBS}$-diagram $\Gamma \vdash D:n$, a polarisation $c\in \hv$ and a position $p\in [n]$, let $(D,c,p)\xRightarrow w (c',p')$ with $w\in \Gamma^*$ a word over $\Gamma$ (or just $(D,c,p)\Rightarrow(c',p')$ for the empty word {$w = \epsilon$}) be inductively defined as follows:
\[
\left(\begin{tikzpicture}
	\begin{pgfonlayer}{nodelayer}
		\node [style=none] (0) at (-0.3, 0) {};
		\node [style=none] (1) at (0.3, 0) {};
	\end{pgfonlayer}
	\begin{pgfonlayer}{edgelayer}
		\draw [style=noire] (0.center) to (1.center);
	\end{pgfonlayer}
\end{tikzpicture},c,0\right)\Rightarrow (c,0)
\qquad
\left(\begin{tikzpicture}[scale=0.8]
	\begin{pgfonlayer}{nodelayer}
		\node [style=cartouche] (0) at (0, 0) {\tiny$\neg$};
		\node [style=none] (1) at (-0.8, 0) {};
		\node [style=none] (2) at (0.8, 0) {};
	\end{pgfonlayer}
	\begin{pgfonlayer}{edgelayer}
		\draw [style=noire] (1.center) to (0);
		\draw [style=noire] (0) to (2.center);
	\end{pgfonlayer}
\end{tikzpicture}
,\uparrow,0\right)\Rightarrow(\rightarrow,0)
\qquad
\left(\begin{tikzpicture}[scale=0.8]
	\begin{pgfonlayer}{nodelayer}
		\node [style=cartouche] (0) at (0, 0) {\tiny$\neg$};
		\node [style=none] (1) at (-0.8, 0) {};
		\node [style=none] (2) at (0.8, 0) {};
	\end{pgfonlayer}
	\begin{pgfonlayer}{edgelayer}
		\draw [style=noire] (1.center) to (0);
		\draw [style=noire] (0) to (2.center);
	\end{pgfonlayer}
\end{tikzpicture}
,\rightarrow,0\right)\Rightarrow(\uparrow,0)
\]
\[
\left(\input{swap-xs.tikz},\,c,\,p\right)\Rightarrow(c,1-p)
\qquad
\left(\input{beamsplitter-xs.tikz},\rightarrow,p\right)\Rightarrow(\rightarrow,p)
\qquad
\left(\input{beamsplitter-xs.tikz},\,\uparrow\,,p\right)\Rightarrow(\uparrow,1-p)
\]
\[
\left(\begin{tikzpicture}[scale=0.8]
	\begin{pgfonlayer}{nodelayer}
		\node [style=newe] (0) at (0, 0) {\phantom{$a$}};
		\node [style=none] (1) at (-1.2, 0) {};
		\node [style=none] (2) at (1.2, 0) {};
		\node [style=none] (3) at (0.8, -0.35) {\footnotesize$a$};
	\end{pgfonlayer}
	\begin{pgfonlayer}{edgelayer}
		\draw [style=new] (1.center) to (2.center);
	\end{pgfonlayer}
\end{tikzpicture}
,c,0\right)\xRightarrow{a} (c,0)
\qquad
\dfrac{(D_1,c,p)\xRightarrow{w_1}(c',p')\qquad(D_2,c',p')\xRightarrow{w_2}(c'',p'')}{(D_2\circ D_1,c,p)\xRightarrow{w_1w_2} (c'',p'')}(\circ)
\]
\[
\dfrac{D_1:n_1\quad p<n_1\quad (D_1,c,p)\xRightarrow{w}(c',p')}{(D_1\oplus D_2,c,p)\xRightarrow{w}(c',p')}(\oplus_1)
\quad
\dfrac{D_1:n_1\quad p\geq n_1\quad (D_2,c,p{-}n_1)\xRightarrow{w}(c',p')}{(D_1\oplus D_2,c,p)\xRightarrow{w}(c',p'{+}n)}(\oplus_2)
\]
\[
\dfrac{D:n+1\qquad \forall i\in\{0,\ldots,k\}, (D,c_i,p_i)\xRightarrow{w_i\,}(c_{i+1},p_{i+1})\qquad(p_{i+1}=n){\Leftrightarrow}(i<k)}{(Tr(D),c_0,p_0)\xRightarrow{w_0\cdots w_k\,}(c_{k+1},p_{k+1})}(\mathsf T_k)
\]
  \noindent with $k = 0$, $1$, and $2$.
  
We denote by  $w^D_{c,p} \in \Gamma^*$  the word, $c^D_{c,p}\in \{\uparrow, \rightarrow\}$ the polarisation, and $p^D_{c,p}\in [n]$ the position s.t. $(D,c,p)\xRightarrow{w^D_{c,p}} (c^D_{c,p},p^D_{c,p})$.

\end{definition}
  
\begin{figure}
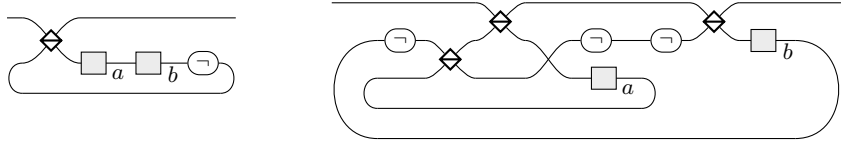

\centerline{$\input{diagrammesimpleabab.tikz}\quad\quad\quad\input{diagrammecompliqueabab.tikz}$}
\caption{\label{fig:barePBS}Two examples of bare \textup{PBS}-diagrams, with the same word path semantics: $(D,\uparrow,0)\xRightarrow{abab}(\uparrow,0)$ and $(D,\rightarrow,0)\xRightarrow{\epsilon}(\rightarrow,0)$. }
\end{figure}

The word path semantics is invariant modulo structural congruence (i.e.,~diagram deformation). Moreover, note that despite the traces which form feedback loops, the word path semantics is well-defined.%
\footnote{Definition~\ref{wordpathsemantics} does not provide any word path semantics for diagrams of type $D:0$. In fact, no word path semantics needs to be defined for such diagrams, as there is no position $p$ defining any input wire. Note also that for diagrams $D:n$ containing fully closed subdiagrams (e.g., of the form $D = D_1 \oplus D_2$ with $D_2:0$), the semantics does not depend on these fully closed subdiagrams.} Indeed, a particle entering the diagram through some input wire cannot go through a feedback loop (or any other part of the diagram) twice with the same polarisation, which justifies that $k$ only needs to go up to $2$ in Rule $(\mathsf T_k)$ above.
Intuitively, if a particle goes twice in a feedback loop with the same polarisation then it will loop forever; but because of time symmetry this also means that the particle went though the feedback loop infinitely many times in the past, which contradicts the fact that it entered through an input wire.
See \cref{appendix:semanticsdeform} for details about the formal proofs of these facts.

For similar reasons, each gate cannot appear more than twice along any path, or even in the family of all the possible paths of a diagram: 

\begin{restatable}{proposition}{proppasplusdedeuxfois}
\label{pasplusdedeuxfois}
Given a bare \textup{PBS}-diagram ${\Gamma} \vdash D:n$, $\forall a\in\Gamma$, one has $\sum\limits_{c\in\{\rightarrow,\uparrow\}, p\in [n]} |w_{c,p}^D|_a\le 2$, where $|w|_a$ denotes the number of occurrences of $a$ in the word $w$. Moreover, if $D$ is $\begin{tikzpicture}[scale=0.8]
	\begin{pgfonlayer}{nodelayer}
		\node [style=cartouche] (0) at (0, 0) {\tiny$\neg$};
		\node [style=none] (1) at (-0.8, 0) {};
		\node [style=none] (2) at (0.8, 0) {};
	\end{pgfonlayer}
	\begin{pgfonlayer}{edgelayer}
		\draw [style=noire] (1.center) to (0);
		\draw [style=noire] (0) to (2.center);
	\end{pgfonlayer}
\end{tikzpicture}
$-free then for any $c$ one has $\sum_{p\in[n]}|w^D_{c,p}|_a\leq1$.
\end{restatable}

The converse is also true: 
\begin{restatable}{proposition}{proprecipdeuxfois}
\label{recip2fois}
For any family of words $\{w_{c,p}\}_{(c,p)\in \{\rightarrow,\uparrow\} \times [n]}$ such that every letter appears at most twice  in the whole family, there exists a bare \textup{PBS}-diagram $D:n$ such that $w_{c,p}=w^D_{c,p}$ for all $c,p$. Furthermore if for any $c\in \{\rightarrow,\uparrow\}$, every letter appears at most once in $\{w_{c,p}\}_{p\in [n]}$, the bare \textup{PBS}-diagram $D$ can be chosen 
$\begin{tikzpicture}[scale=0.8]
	\begin{pgfonlayer}{nodelayer}
		\node [style=cartouche] (0) at (0, 0) {\tiny$\neg$};
		\node [style=none] (1) at (-0.8, 0) {};
		\node [style=none] (2) at (0.8, 0) {};
	\end{pgfonlayer}
	\begin{pgfonlayer}{edgelayer}
		\draw [style=noire] (1.center) to (0);
		\draw [style=noire] (0) to (2.center);
	\end{pgfonlayer}
\end{tikzpicture}
$-free.
\end{restatable}

The proofs are given in Appendices~\ref{app:pasplusdedeuxfois} and~\ref{app:recip2fois}.
Note in particular that the proof of Proposition~\ref{recip2fois} is constructive. For instance, the family $\{w_{\uparrow,0}=abab, w_{\rightarrow,0}=\epsilon\}$ can be obtained from the diagram of Fig.~\ref{fig:barePBS}~(Right). The solution is not unique in general and there is actually a simpler diagram, see Fig.~\ref{fig:barePBS}~(Left),  with the same word path semantics.

\subsection{Extended PBS-diagrams}
\label{subsec:extendedPBS}

We will now introduce extended \textup{PBS}-diagrams by filling every bare gate  with the description of a quantum channel. As recalled in the introduction, however, defining the coherent control of general channels (as we wish to do with PBS-diagrams) in an unambiguous way is not trivial. Here we propose to do so through the notion of purified channels, which are an extension of Stinespring's dilation of quantum channels~\cite{Stinespring55Positive}.

\subsubsection{Purified channels}
\label{sec:purif} 

A standard paradigm for quantum channels acting on a Hilbert space $\H$ is to describe them as CPTP maps, or superoperators  $\L(\H)\to\L(\H)$,%
\footnote{As this is the case of interest in PBS-diagrams (with $\H$ corresponding to the data register), we consider here channels with the same input and output Hilbert spaces.}
 where $\L(\H)$ denotes the set of linear operators on $\H$. As exemplified e.g. in~\cite{oi03,Abbott2020communication}, this representation is however ambiguous when it comes to describing quantum coherent control: two quantum channels with the same superoperator can behave differently in a coherent-control setting.

A possible way to overcome this issue is to ``go to the Church of the larger Hilbert space'', according to which any quantum channel can be interpreted as a pure quantum operation acting on both the quantum system and an environment. Mathematically, this corresponds to Stinespring's dilation theorem~\cite{Stinespring55Positive}, which states that any CPTP map acting on a Hilbert space $\H$ can be implemented with an isometry $V:\H\to \H\otimes \E$, where $\E$ denotes the Hilbert space attached to the environment, followed by a partial trace of the latter. Note that in this representation, the isometry $V$ can be understood as encoding both the creation of the environment $\E$ and the evolution of the joint system $\H\otimes \E$. Indeed, $V$ can always be decomposed into an environment initialisation $\ket {\varepsilon}\in \E$ and a unitary evolution $U:\H\otimes \E\to \H\otimes \E$ such that $V= U (I_\H\otimes \ket{\varepsilon})$, where $I_\H$ denotes the identity operator over $\H$.
In our approach to defining coherent control for quantum channels, we will precisely abide by this description in terms of unitary purifications, which we formalise as follows:

\begin{definition}[Purified channel]
Given a Hilbert space $\H$, a \emph{purified $\H$-channel} (or simply \emph{purified channel}, for short) is a triplet $[U,\ket \varepsilon, \E]$, where $\E$ is the local environment Hilbert space, $\ket \varepsilon \in \E$ is the environment initial state, and $U: \H\otimes \E \to \H\otimes \E$ is a unitary operator representing the evolution of the joint system.
We denote the set of purified ${\mathcal H}$-channels by $\mathfrak{C}({\mathcal H})$.
\end{definition}

As seen above, it directly follows from Stinespring's dilation theorem that any CPTP map $\L(\H) \to \L(\H)$ can be represented by a purified ${\mathcal H}$-channel, which is however not unique. Reciprocally, with any  purified $\H$-channel $[U,\ket \varepsilon, \E]$, we naturally associate the CPTP map 
$\S^{(1)}_{[U,\ket{\varepsilon},\E]}: \L(\H)\to\L(\H) = \rho \mapsto \textup{Tr}_{\E}\big(U (\rho \otimes \ket{\varepsilon}\!\bra{\varepsilon}) U^\dagger \big)$, where $\textup{Tr}_{\E}$ denotes the partial trace over $\E$, and which we shall represent graphically, using the circuit notations of Appendix~\ref{appendix_circuit_notations},%
\footnote{To manipulate unitary operations  and CPTP maps, it is convenient to use such circuit-like graphical representations, which correspond to standard circuit notations for ``pure'' operations, supplemented with a ground symbol $\input{ground.tikz}$ for the case of CPTP maps; see Appendix~\ref{appendix_circuit_notations} for details.}
as follows: $\S^{(1)}_{[U,\ket{\varepsilon},\E]} = \input{FU_D_p.tikz}$. 

One may however not trace out the environment straight away. In fact, decomposing Stinespring's dilation into an environment state initialisation and a unitary evolution of the joint system, as we did above, allows one to apply the same channel several times in a coherent manner if a particle goes through a gate several times. In that case we will consider that the same unitary is applied each time, without re-initialising the environment state (which we assume to not evolve between two applications of the channel).

\subsubsection{From bare to extended PBS-diagrams}

We are now in a position to define extended \textup{PBS}-diagrams of type $\mathcal H^{(n)}$, which are essentially bare PBS-diagrams of type $n$, where the gate indices are replaced by purified $\H$-channels.
Hence, instead of bare gates $\begin{tikzpicture}[scale=0.8]
	\begin{pgfonlayer}{nodelayer}
		\node [style=newe] (0) at (0, 0) {\phantom{$a$}};
		\node [style=none] (1) at (-1.2, 0) {};
		\node [style=none] (2) at (1.2, 0) {};
		\node [style=none] (3) at (0.8, -0.35) {\footnotesize$a$};
	\end{pgfonlayer}
	\begin{pgfonlayer}{edgelayer}
		\draw [style=new] (1.center) to (2.center);
	\end{pgfonlayer}
\end{tikzpicture}
$, an extended \textup{PBS}-diagram contains gates of the form $\begin{tikzpicture}[scale=0.8]
	\begin{pgfonlayer}{nodelayer}
		\node [style=newe] (0) at (0, 0) {\footnotesize $U,\ket{\varepsilon}$};
		\node [style=none] (1) at (-1.6, 0) {};
		\node [style=none] (2) at (1.6, 0) {};
	\end{pgfonlayer}
	\begin{pgfonlayer}{edgelayer}
		\draw [style=new] (1.center) to (2.center);
	\end{pgfonlayer}
\end{tikzpicture}$, parametrised by a purified channel $[U,\ket \varepsilon,\E] \in \mathfrak{C}({\mathcal H})$ (where the Hilbert space $\E$ is not represented explicitly, in order not to overload the diagrams). 

This leads to the following inductive definition:

\begin{definition}[Extended PBS-diagram]
 \label{def:extendedPBS-diagrams}
An extended   \textup{PBS}-diagram $  D:\typ{\Gamma}n$ (with $n \in{ \mathbb N}$) is inductively defined as:
\[
\begin{tikzpicture}[scale=0.3]
	\begin{pgfonlayer}{nodelayer}
		\node [style=sdiagrammevide] (0) at (0, 0) {};
	\end{pgfonlayer}
\end{tikzpicture}
\!:\!\typ{.}0\qquad 
\begin{tikzpicture}
	\begin{pgfonlayer}{nodelayer}
		\node [style=none] (0) at (-0.3, 0) {};
		\node [style=none] (1) at (0.3, 0) {};
	\end{pgfonlayer}
	\begin{pgfonlayer}{edgelayer}
		\draw [style=noire] (0.center) to (1.center);
	\end{pgfonlayer}
\end{tikzpicture}\!:\!\typ{.}1\qquad
\begin{tikzpicture}[scale=0.8]
	\begin{pgfonlayer}{nodelayer}
		\node [style=cartouche] (0) at (0, 0) {\tiny$\neg$};
		\node [style=none] (1) at (-0.8, 0) {};
		\node [style=none] (2) at (0.8, 0) {};
	\end{pgfonlayer}
	\begin{pgfonlayer}{edgelayer}
		\draw [style=noire] (1.center) to (0);
		\draw [style=noire] (0) to (2.center);
	\end{pgfonlayer}
\end{tikzpicture}
\!:\!\typ{.}1\qquad
\input{swap-xs.tikz}\!:\!\typ{.}2\qquad
\input{beamsplitter-xs.tikz}\!:\!\typ{.}2 \qquad
\dfrac{[U,\ket \varepsilon,\E] \in \mathfrak{C}({\mathcal H})}{
\begin{tikzpicture}[scale=0.8]
	\begin{pgfonlayer}{nodelayer}
		\node [style=newe] (0) at (0, 0) {\footnotesize $U,\ket{\varepsilon}$};
		\node [style=none] (1) at (-1.6, 0) {};
		\node [style=none] (2) at (1.6, 0) {};
	\end{pgfonlayer}
	\begin{pgfonlayer}{edgelayer}
		\draw [style=new] (1.center) to (2.center);
	\end{pgfonlayer}
\end{tikzpicture}:\typ{\{x\}}{1}}
\]
\[
\dfrac{
D_1:\typ{\Gamma_1}n \quad
D_2 : \typ{\Gamma_2}n}{
D_2\circ D_1:\typ{\Gamma_1\cup\Gamma_2}{n}}\qquad
\dfrac{
D_1:\typ{\Gamma_1}{n_1}  \quad
D_2 : \typ{\Gamma_2}{n_2}
}{
D_1\oplus D_2:\typ{\Gamma_1\cup\Gamma_2}{n_1+n_2}}\qquad
 \dfrac{
D:\typ{\Gamma}{n+1}}{
Tr(D):\typ{\Gamma}{n}}
\]
\end{definition}

Extended PBS-diagrams are defined up to the same structural congruence as for bare PBS-diagrams.
It is convenient to explicitly define the map which, given a family of purified channels, transforms a bare diagram into the corresponding extended PBS-diagram:%
\footnote{To clarify which kind of diagram we are dealing with, in this subsection we use primed names (e.g., $D'$) when referring to bare PBS-diagrams, and nonprimed names for extended PBS-diagrams.}

\begin{definition}\label{def:map_bare_to_extended}
Given a bare \textup{PBS}-diagram $\Gamma\vdash D':n$ and a family of purified $\H$-channels $\mathcal G = ([U_a, \ket {\varepsilon_a}, \E_a])_{a\in \Gamma}$  indexed by elements of  $\Gamma$, let $[D']_{\mathcal G}:\mathcal \H^{(n)}$ be the extended \textup{PBS}-diagram inductively defined as $[\begin{tikzpicture}[scale=0.8]
	\begin{pgfonlayer}{nodelayer}
		\node [style=newe] (0) at (0, 0) {\phantom{$a$}};
		\node [style=none] (1) at (-1.2, 0) {};
		\node [style=none] (2) at (1.2, 0) {};
		\node [style=none] (3) at (0.8, -0.35) {\footnotesize$a$};
	\end{pgfonlayer}
	\begin{pgfonlayer}{edgelayer}
		\draw [style=new] (1.center) to (2.center);
	\end{pgfonlayer}
\end{tikzpicture}
]_{([U_a, \ket {\varepsilon_a}, \E_a])} = \begin{tikzpicture}[scale=0.8]
	\begin{pgfonlayer}{nodelayer}
		\node [style=newe] (0) at (0, 0) {\footnotesize $U_a,\ket{\varepsilon_a}$};
		\node [style=none] (1) at (-2, 0) {};
		\node [style=none] (2) at (2, 0) {};
	\end{pgfonlayer}
	\begin{pgfonlayer}{edgelayer}
		\draw [style=new] (1.center) to (2.center);
	\end{pgfonlayer}
\end{tikzpicture}
$, $\forall g\in \{\begin{tikzpicture}[scale=0.3]
	\begin{pgfonlayer}{nodelayer}
		\node [style=sdiagrammevide] (0) at (0, 0) {};
	\end{pgfonlayer}
\end{tikzpicture}
,
\begin{tikzpicture}
	\begin{pgfonlayer}{nodelayer}
		\node [style=none] (0) at (-0.3, 0) {};
		\node [style=none] (1) at (0.3, 0) {};
	\end{pgfonlayer}
	\begin{pgfonlayer}{edgelayer}
		\draw [style=noire] (0.center) to (1.center);
	\end{pgfonlayer}
\end{tikzpicture},
\begin{tikzpicture}[scale=0.8]
	\begin{pgfonlayer}{nodelayer}
		\node [style=cartouche] (0) at (0, 0) {\tiny$\neg$};
		\node [style=none] (1) at (-0.8, 0) {};
		\node [style=none] (2) at (0.8, 0) {};
	\end{pgfonlayer}
	\begin{pgfonlayer}{edgelayer}
		\draw [style=noire] (1.center) to (0);
		\draw [style=noire] (0) to (2.center);
	\end{pgfonlayer}
\end{tikzpicture}
,
\input{swap-xs.tikz},
\input{beamsplitter-xs.tikz}\}$, $[g]_{\emptyset}=g$, $[D_2'\circ  D_1']_{\mathcal G_1\uplus\mathcal G_2} = [D_2']_{\mathcal G_2}\circ [D_1']_{\mathcal G_1}$, $[D_1'\oplus D_2']_{\mathcal G_1\uplus\mathcal G_2} = [D_1']_{\mathcal G_1}\oplus [D_2']_{\mathcal G_2}$  and $[Tr(D')]_{\mathcal G} = Tr([D']_{\mathcal G})$, where $\uplus$ is the disjoint union. 
\end{definition}

For any extended \textup{PBS}-diagram $D:\H^{(n)}$, there exists a bare diagram $\Gamma\vdash D':n$ and an indexed family of purified $\H$-channels  $\mathcal G$ s.t. $[D']_{\mathcal G} = D$. We call $D'$ an \emph{underlying bare diagram} of $D$ (which is unique, up to relabelling of the gates).

\subsubsection{Quantum semantics}

We now equip the extended PBS-diagrams with a quantum semantics, which is a CPTP map acting on the complete state of the particle that goes through it,~i.e.,~its joint polarisation, position and data state. To describe the quantum semantics of an extended PBS-diagram {$D:\typ{}n$}, it is convenient to rely on an underlying bare diagram $\Gamma\vdash D':n$ and a family of purified channels $\mathcal G$ s.t. $[D']_{\mathcal G} = D$ (so as to keep track of the environment spaces and be able to identify them via the bare gate indices). 
 
As we defined them, every purified channel comes with its local environment and a unitary evolution acting on both the data register and its local environment.  In order to define the overall evolution of the diagram, we consider the global environment as the tensor product of these local environments, and extend every unitary transformation to a global  transformation acting on the data register and the global environment:

\begin{definition}
Given  an indexed family of purified $\H$-channels  $\mathcal G = ([U_a, \ket {\varepsilon_a}, \E_a])_{a\in \Gamma}$, let $\E_{\mathcal G} \coloneqq \bigotimes_{a\in\Gamma} \E_a$, $\ket{\varepsilon_{\mathcal G}} \coloneqq \bigotimes_{a\in \Gamma}  \ket{\varepsilon_a} \in \E_{\mathcal G} $, and $\forall\,a\in \Gamma$, let $V_a^{\mathcal G} \coloneqq U_a \bigotimes_{x\in\Gamma\backslash\{a\}} I_{\E_x} \in \L(\mathcal H\otimes \mathcal E_{\mathcal G})$. 
\end{definition}

If a particle enters an extended PBS-diagram $D$ with a definite polarisation and position in some basis states $\ket{c} \in \CC^{\hv}$ and $\ket{p} \in \CC^{[n]}$, respectively, the sequence of transformations applied to the particle and the global environment when the particle goes through the diagram can be deduced from the word path semantics of the underlying bare diagram $D'$:
\[
\ket{c} \otimes \ket{p} \otimes \ket{\psi} \otimes \ket{\varepsilon_{\mathcal G}} \mapsto \ket{c^{D'}_{c,p}} \otimes \ket{p^{D'}_{c,p}} \otimes V^{\mathcal G}_{w^{D'}_{c,p}} (\ket{\psi} \otimes \ket{\varepsilon_{\mathcal G}})
\]
where $w^{D'}_{c,p}$, $c^{D'}_{c,p}$, and $p^{D'}_{c,p}$ are given by the word path semantics, i.e., $(D',c,p)\xRightarrow{w^{D'}_{c,p}} (c^{D'}_{c,p},p^{D'}_{c,p})$, and $V^{\mathcal G}_{w}$ is inductively defined as $V^{\mathcal G}_{\epsilon} \coloneqq I_{\H\otimes \E}$ and $\forall a\in \Gamma, \forall w\in \Gamma^*$, $V^{\mathcal G}_{aw}\coloneqq V^{\mathcal G}_{w}V^{\mathcal G}_{a}$.

One can actually consider inputting a particle in an arbitrary initial state (i.e., including superpositions of polarisation and position); the transformation applied by the diagram is then obtained from the one above, by linearity.
This leads us to define the following:
\begin{definition}
Given a bare \textup{PBS}-diagram $\Gamma \vdash D':n$ and a family of purified $\H$-channels $\mathcal G$ indexed with $\Gamma$,  let $$U_{D'}^{\mathcal G} \coloneqq \sum_{c\in\{\rightarrow,\uparrow\}, p\in [n]} \ket{c_{c,p}^{D'}}\bra{c}\otimes\ket{p_{c,p}^{D'}}\bra{p} \otimes V^{\mathcal G}_{w^{D'}_{c,p}}$$
\end{definition}

The triplet $[U_{D'}^{\mathcal G},\ket{\varepsilon_{\mathcal G}}, \E_{\mathcal G}]$ is nothing but a purified $(\CC^{\hv}\otimes\CC^{[n]}\otimes\H)$-channel, which describes the action of the corresponding extended PBS-diagram on the complete state of the particle.
Once the particle exits the diagram, the environments of all purified channels are not accessible anymore. 
As is well-known, the statistics of any ``input/output test'', which consists in preparing an arbitrary input state of the particle and measuring the output in an arbitrary basis, then only depend on the CPTP map (the superoperator) induced by $U_{D'}^{\mathcal G}$ above, with all environments initially prepared in the global state $\ket{\varepsilon_{\mathcal G}}$, and after tracing out all environment spaces---i.e., using circuit-like notations: $\input{FU_DG.tikz}\, $. This superoperator thus precisely captures input/output (in)distinguishability: two quantum channels have the same superoperator if and only if they are indistinguishable in any input/output test.
This provides the ground for our definition of the following quantum semantics:

\begin{definition}[Quantum Semantics]
\label{def:quantum_semantics}
Given an extended $\textup{PBS}$-diagram $D:\typ{}n$, let $\interp D: \L(\CC^{\hv}\otimes\CC^{[n]}\otimes\H) \to \L(\CC^{\hv}\otimes\CC^{[n]}\otimes\H)$ be the superoperator defined as $$ \interp D \ \coloneqq \ \rho \mapsto \textup{Tr}_{\mathcal E_{\G}} (U_{D'}^{\mathcal G}(\rho\otimes \ket{\varepsilon_{\mathcal G}}\bra{\varepsilon_{\mathcal G}}){U^{\mathcal G}_{D'}}^\dagger) \quad = \quad \input{FU_DG.tikz}$$ where $\Gamma \vdash D':n$ is an underlying bare diagram and $\mathcal G$ is an indexed family of purified $\H$-channels s.t. $[D']_{\mathcal G}= D$.
\end{definition}

Note that the quantum semantics is preserved by the `only topology matters' structural congruence on diagrams. Indeed, it is defined using only the family $\G$ and the word path semantics of its underlying bare diagram $D'$, which is invariant modulo diagram deformation. It is clear that when deforming $D$ we do not have to change $D'$ and $\G$, since it suffices to deform $D'$ accordingly.

\section{Observational equivalence of purified channels}
\label{sec:obs_equiv}

In this section we address the problem of deciding whether  two  purified channels $[U,\ket{\varepsilon},\E]$ and $[U',\ket{\varepsilon'},\E']$ can be distinguished in an experiment involving coherent control, within the framework of PBS-diagrams just established.
We introduce for that the notion of \emph{contexts}, which are extended PBS-diagrams with a ``hole'':
 if for  any context, filling its hole with $[U,\ket{\varepsilon},\E]$ or $[U',\ket{\varepsilon'},\E']$ leads to diagrams with the same quantum semantics, then the two purified channels $[U,\ket{\varepsilon},\E]$ and $[U',\ket{\varepsilon'},\E']$ are indistinguishable within our framework, even with the help of the coherent control provided by extended PBS-diagrams.

\subsection{Contexts}

A context is an extended \textup{PBS}-diagram with a \emph{hole}, i.e., a (unique) particular empty gate, without any purified channel specified \emph{a priori}. Equivalently a context can be seen as a bare \textup{PBS}-diagram partially filled: all but one gate are filled with purified channels. Formally:
\begin{definition}[Context]\label{def:context}
A context $C[\cdot] \!:\!\typ {\Gamma}n$ (with $n \!\in{ \!\mathbb N}$) is inductively defined as follows: 
\begin{itemize} 
\item The hole gate $\begin{tikzpicture}[scale=0.8]
	\begin{pgfonlayer}{nodelayer}
		\node [style=newe] (0) at (0, 0) {$\cdot$};
		\node [style=none] (1) at (-0.95, 0) {};
		\node [style=none] (2) at (0.95, 0) {};
	\end{pgfonlayer}
	\begin{pgfonlayer}{edgelayer}
		\draw [style=new] (1.center) to (2.center);
	\end{pgfonlayer}
\end{tikzpicture}
:\typ {\Gamma}1$ is a context;
\item If $C[\cdot]:\typ {\Gamma}n$ is a context and $D:\typ {\Gamma}n$ is an extended \textup{PBS}-diagram then
$ D\circ C[\cdot]:\typ {\Gamma}n$ and $C[\cdot]\circ D:\typ {\Gamma}n$ are contexts;
\item If $C[\cdot]:\typ {\Gamma}n$ is a context and $D:\typ {\Gamma}m$ is an extended \textup{PBS}-diagram then
$D \oplus C[\cdot]:\typ {\Gamma}{m+n}$ and $C[\cdot] \oplus D:\typ {\Gamma}{n+m}$ are contexts;
\item If $C[\cdot]:\typ {\Gamma}{n+1}$ is a context then $Tr(C[\cdot]) :\typ {\Gamma}{n}$ is a context.
\end{itemize}
\end{definition}

Like bare and extended PBS-diagrams, contexts are defined up to structural congruence.

\begin{definition}[Substitution]\label{def:substitution}
For any context $C[\cdot]:\typ{\Gamma}n$ and any purified $\H$-channel $[U,\ket{\varepsilon},\mathcal E]$, let $C[U,\ket {\varepsilon},\E]:\typ{\Gamma}n$ be the extended \textup{PBS}-diagram obtained by replacing the single hole $\begin{tikzpicture}[scale=0.8]
	\begin{pgfonlayer}{nodelayer}
		\node [style=newe] (0) at (0, 0) {$\cdot$};
		\node [style=none] (1) at (-0.95, 0) {};
		\node [style=none] (2) at (0.95, 0) {};
	\end{pgfonlayer}
	\begin{pgfonlayer}{edgelayer}
		\draw [style=new] (1.center) to (2.center);
	\end{pgfonlayer}
\end{tikzpicture}
$ in $C[\cdot]$ by the purified channel $\begin{tikzpicture}[scale=0.8]
	\begin{pgfonlayer}{nodelayer}
		\node [style=newe] (0) at (0, 0) {\footnotesize $U,\ket{\varepsilon}$};
		\node [style=none] (1) at (-1.6, 0) {};
		\node [style=none] (2) at (1.6, 0) {};
	\end{pgfonlayer}
	\begin{pgfonlayer}{edgelayer}
		\draw [style=new] (1.center) to (2.center);
	\end{pgfonlayer}
\end{tikzpicture}$.
\end{definition}

After some purified channel is plugged in, contexts allow one to compare the quantum semantics $\interp{C[U, \ket \varepsilon,\E]}$ and $ \interp{C[U',\ket{\varepsilon'},\E']}$ induced by different purified channels $[U,\ket{\varepsilon},\E]$ and $[U',\ket{\varepsilon'},\E']$.
We consider in the following three subclasses of contexts, depending on the kind of coherent control one may allow to distinguish purified channels: whether we exclude the use of PBS ($\input{beamsplitter-xs.tikz}$), of polarisation flips (``negations'' $\begin{tikzpicture}[scale=0.8]
	\begin{pgfonlayer}{nodelayer}
		\node [style=cartouche] (0) at (0, 0) {\tiny$\neg$};
		\node [style=none] (1) at (-0.8, 0) {};
		\node [style=none] (2) at (0.8, 0) {};
	\end{pgfonlayer}
	\begin{pgfonlayer}{edgelayer}
		\draw [style=noire] (1.center) to (0);
		\draw [style=noire] (0) to (2.center);
	\end{pgfonlayer}
\end{tikzpicture}
$), or whether we allow both.  
This leads us to define the following equivalence relations:

\begin{definition} [Observational equivalences]
\label{def_obs_equiv}
Given two purified $\H$-channels $[U,\ket{\varepsilon},\E]$ and $[U',\ket{\varepsilon'},\E']$, we consider the three following refinements of observational equivalences (for $i\in\{0,1,2\}$): 
 ${[U,\ket{\varepsilon},\E] \approx_i [U',\ket{\varepsilon'},\E']}$ if $\forall C[\cdot]\in \mathcal C_i$, $\interp{C[U, \ket \varepsilon,\E]} = \interp{C[U',\ket{\varepsilon'},\E']}$, where:
 \begin{itemize}
 \item $\mathcal C_0$ is the set of $\input{beamsplitter-xs.tikz}$-free contexts $C[\cdot]:\mathcal H^{(1)}$;
 \item $\mathcal C_1$ is the set of $\begin{tikzpicture}[scale=0.8]
	\begin{pgfonlayer}{nodelayer}
		\node [style=cartouche] (0) at (0, 0) {\tiny$\neg$};
		\node [style=none] (1) at (-0.8, 0) {};
		\node [style=none] (2) at (0.8, 0) {};
	\end{pgfonlayer}
	\begin{pgfonlayer}{edgelayer}
		\draw [style=noire] (1.center) to (0);
		\draw [style=noire] (0) to (2.center);
	\end{pgfonlayer}
\end{tikzpicture}
$-free contexts $C[\cdot]:\mathcal H^{(1)}$;
 \item $\mathcal C_2$ is the set of all contexts $C[\cdot]:\mathcal H^{(1)}$.
 \end{itemize}
\end{definition}

Note that contexts in $\mathcal C_0$ do not perform any coherent control; these consist in just a linear sequence of gates and negations, possibly composed in parallel with closed loops (i.e., traces of such sequences), including a hole gate somewhere.
It is clear, by deformation of diagrams, that more general contexts can always be described as follows: 

\begin{proposition} For any context $C[\cdot]\in \mathcal C_2$  there exists an extended \textup{PBS}-diagram ${D}$ such that  
$C[\cdot]=\input{contextedeformetrou.tikz}$. Moreover if $C[\cdot]\in \mathcal C_1$ then $D$ can be chosen  $\begin{tikzpicture}[scale=0.8]
	\begin{pgfonlayer}{nodelayer}
		\node [style=cartouche] (0) at (0, 0) {\tiny$\neg$};
		\node [style=none] (1) at (-0.8, 0) {};
		\node [style=none] (2) at (0.8, 0) {};
	\end{pgfonlayer}
	\begin{pgfonlayer}{edgelayer}
		\draw [style=noire] (1.center) to (0);
		\draw [style=noire] (0) to (2.center);
	\end{pgfonlayer}
\end{tikzpicture}
$-free.
\end{proposition}

\begin{remark}  
In Definition~\ref{def_obs_equiv} we only consider contexts with a single input/output wire. This is because we intend to use contexts to distinguish purified channels; now, if one can distinguish two purified channels with a context of type $\H^{(n)}$ but no context of type $\H^{(1)}$, then intuitively this means that the extra power comes from the preparation of the initial state and/or some particular measurement, which are not represented in the context. Actually, except in the  $\mathcal C_0$ case, allowing multiple input/output wires does not increase the distinguishability power of the contexts (see \cref{multiwirecontextsnegfree,multiwirecontexts} in \cref{appendix_obs_equiv}).
\end{remark}

\subsection{Observational equivalence using PBS-free contexts}

Let us start by characterising which purified channels are indistinguishable by $\input{beamsplitter-xs.tikz}$-free contexts in $\mathcal C_0$. Not surprisingly, we recover the usual indistinguishability by input/output tests, which is captured by the fact that the two purified channels lead to the same superoperator:\footnote{In other words, if two purified channels can be distinguished using a $\input{beamsplitter-xs.tikz}$-free context, then they could already be distinguished with simply an input/output test (or with a trivial context $\begin{tikzpicture}[scale=0.8]
	\begin{pgfonlayer}{nodelayer}
		\node [style=newe] (0) at (0, 0) {$\cdot$};
		\node [style=none] (1) at (-0.95, 0) {};
		\node [style=none] (2) at (0.95, 0) {};
	\end{pgfonlayer}
	\begin{pgfonlayer}{edgelayer}
		\draw [style=new] (1.center) to (2.center);
	\end{pgfonlayer}
\end{tikzpicture}
$).}

\begin{definition}[(First-level) Superoperator] Given a purified $\H$-channel ${[U,\ket{\varepsilon},\E]}$, let $\S^{(1)}_{[U,\ket{\varepsilon},\E]}: \ \mathcal L(\H) \to \mathcal L(\H) = \rho \mapsto \textup{Tr}_{\E}\big(U (\rho \otimes \ket{\varepsilon}\!\bra{\varepsilon}) U^\dagger \big)$ be the (``first-level'') superoperator of  ${[U,\ket{\varepsilon},\E]}$. Graphically, 
$$\S^{(1)}_{[U,\ket{\varepsilon},\E]} \coloneqq \input{FU_L.tikz}$$
\end{definition}

\begin{restatable}{theorem}{thmcptpchannels} \label{thm:cptpchannels}
Given two purified $\H$-channels ${[U,\ket{\varepsilon},\E]}$ and ${[U',\ket{\varepsilon'},\E']}$, ${[U,\ket{\varepsilon},\E]} \approx_0 {[U',\ket{\varepsilon'},\E']}$ iff they have the same (first-level) superoperator.
 Graphically,
 \begin{equation}
\hspace{15mm} {[U,\ket{\varepsilon},\E]} \approx_0 {[U',\ket{\varepsilon'},\E']} \qquad\ \text{iff} \qquad\quad \input{FU_L.tikz}\quad = \quad\input{FU_R.tikz} \tag{S1}
\end{equation}
\end{restatable}

The proof is given in Appendix \ref{appendix:PBSfree}.

\subsection{Observational equivalence using negation-free contexts}

Allowing contexts with PBS significantly increases their power to distinguish purified channels.
In~\cite{Abbott2020communication}, a particular kind of coherent control---namely, the \emph{``first half of a quantum switch''}~\cite{chiribella13,araujo14,goswami18}---has been considered, which can be rephrased using contexts of the form: 
$$\input{half_QS.tikz}$$
The authors proved that with these particular contexts, two purified channels leading to the same (first-level) superoperator are indistinguishable if and only if they also have the same \emph{(first-level) transformation matrix}, which is defined as follows%
\footnote{Originally, in~\cite{Abbott2020communication}, the transformation matrix was defined for a given unitary purification of a CPTP map $\mathcal S: \L(\H) \to \L(\H)$ in the form 
 $U: \ket{\psi}_\H\otimes\ket{\varepsilon} \mapsto \sum_i K_i \ket{\psi}_\H\otimes\ket{i}_\E$
  (where the $K_i$'s are Kraus operators of $\mathcal S$, and where an environment space $\E$ was introduced, with an orthonormal basis $\{\ket{i}_\E\}_i$ and an initial state $\ket{\varepsilon}$),
   as $T \coloneqq \sum_i \bra{\varepsilon}i\rangle_\E \, K_i$. This is indeed consistent with our Definition~\ref{def_T1} here, as with these notations $U (I_\H\otimes \ket{\varepsilon}) = \sum_i K_i \otimes\ket{i}_\E$, so that $(I_\H\otimes \bra{\varepsilon}) U (I_\H\otimes \ket{\varepsilon}) = \sum_i \bra{\varepsilon}i\rangle_\E \, K_i = T$.}
   
\begin{definition}[(First-level) Transformation Matrix] \label{def_T1} Given a purified $\H$-channel ${[U,\ket{\varepsilon},\E]}$, let $T^{(1)}_{[U,\ket{\varepsilon},\E]} \coloneqq (I_\H\otimes \bra{\varepsilon}) U (I_\H\otimes \ket{\varepsilon})  \in \mathcal L(\H)$ be the (``first-level'') transformation matrix of  ${[U,\ket{\varepsilon},\E]}$. Graphically, 
$$T^{(1)}_{[U,\ket{\varepsilon},\E]} \coloneqq \input{TM1_L.tikz}$$
\end{definition}

We extend this result to any  $\begin{tikzpicture}[scale=0.8]
	\begin{pgfonlayer}{nodelayer}
		\node [style=cartouche] (0) at (0, 0) {\tiny$\neg$};
		\node [style=none] (1) at (-0.8, 0) {};
		\node [style=none] (2) at (0.8, 0) {};
	\end{pgfonlayer}
	\begin{pgfonlayer}{edgelayer}
		\draw [style=noire] (1.center) to (0);
		\draw [style=noire] (0) to (2.center);
	\end{pgfonlayer}
\end{tikzpicture}
$-free context. 

\begin{restatable}{theorem}{thmcaracobsequivnegfree}\label{caracobsequivnegfree}
Given two purified $\H$-channels ${[U,\ket{\varepsilon},\E]}$ and ${[U',\ket{\varepsilon'},\E']}$, ${[U,\ket{\varepsilon},\E]} \approx_1 {[U',\ket{\varepsilon'},\E']}$ iff
 they have the same (first-level) superoperator and the same (first-level) transformation matrix.
 Graphically,
  \begin{subequations}
    \begin{empheq}[left={{[U,\ket{\varepsilon},\E]} \approx_1 {[U',\ket{\varepsilon'},\E']} \qquad \text{iff} \qquad \empheqlbrace \ }]{align}
      \input{FU_L.tikz}\quad &= \quad\input{FU_R.tikz} \label{SP1} \tag{S1} \\[3mm]
      \input{TM1_L.tikz} \quad &=\quad\input{TM1_R.tikz} \label{TM1_v0} \tag{T1}
    \end{empheq}
  \end{subequations}
\end{restatable}

The proof is given in Appendix~\ref{appendix_obs_equiv_negfree}, and shows at the same time that allowing multiple input/output wires does not increase the power of $\begin{tikzpicture}[scale=0.8]
	\begin{pgfonlayer}{nodelayer}
		\node [style=cartouche] (0) at (0, 0) {\tiny$\neg$};
		\node [style=none] (1) at (-0.8, 0) {};
		\node [style=none] (2) at (0.8, 0) {};
	\end{pgfonlayer}
	\begin{pgfonlayer}{edgelayer}
		\draw [style=noire] (1.center) to (0);
		\draw [style=noire] (0) to (2.center);
	\end{pgfonlayer}
\end{tikzpicture}
$-free contexts.

One can illustrate how the transformation matrices enter the game by considering for example the following context : $\input{boucletraversetrou.tikz}$. 
By plugging in ${[U,\ket{\varepsilon},\E]}$, the extended PBS-diagram maps a pure input state $\frac{\ket{\rightarrow}+\ket{\uparrow}}{\sqrt{2}}\otimes \ket{\psi} \in \CC^{\hv} \otimes \H$ (together with the environment initial state $\ket\varepsilon\in\E$) to the state $\frac{1}{\sqrt{2}}\ket{\rightarrow}\otimes\ket{\psi}\otimes\ket{\varepsilon} + \frac{1}{\sqrt{2}}\ket{\uparrow}\otimes U(\ket{\psi}\otimes\ket{\varepsilon})$, so that after tracing out the environment a cross term $\frac{1}{2}\ket{\uparrow}\!\bra{\rightarrow}\otimes \textup{Tr}_\E \big[ U(\ket{\psi}\!\bra{\psi}\otimes\ket{\varepsilon}\!\bra{\varepsilon}) \big] = \frac{1}{2}\ket{\uparrow}\!\bra{\rightarrow}\otimes T^{(1)}_{[U,\ket{\varepsilon},\E]} \ket{\psi}\!\bra{\psi}$ appears.

We note also that the two conditions~\eqref{SP1} and~\eqref{TM1_v0} are nonredundant, i.e., one does not imply the other. Indeed, there exist cases where $\S^{(1)}_{[U,\ket{\varepsilon},\E]} = \S^{(1)}_{[U',\ket{\varepsilon'},\E']}$ but $T^{(1)}_{[U,\ket{\varepsilon},\E]} \neq T^{(1)}_{[U',\ket{\varepsilon'},\E']}$ (e.g., given any $\H$, $\E=\E'=\CC$, $U = I_\H, U' = -I_\H$ and $\ket{\varepsilon} = \ket{\varepsilon'} = 1$), and cases where $\S^{(1)}_{[U,\ket{\varepsilon},\E]} \neq \S^{(1)}_{[U',\ket{\varepsilon'},\E']}$ but $T^{(1)}_{[U,\ket{\varepsilon},\E]} = T^{(1)}_{[U',\ket{\varepsilon'},\E']}$ (e.g., $\H=\E=\E'=\CC^2$, $U = I_\H \otimes X, U' = X \otimes X \text{ and } \ket{\varepsilon} = \ket{\varepsilon'} = \ket{0}$).\footnote{Where $X=\begin{pmatrix}0&1\\1&0\end{pmatrix}$.}

\subsection{Observational equivalence using general contexts}
\label{subsec:obs_equiv_general}

We will now see that allowing  negations ($\begin{tikzpicture}[scale=0.8]
	\begin{pgfonlayer}{nodelayer}
		\node [style=cartouche] (0) at (0, 0) {\tiny$\neg$};
		\node [style=none] (1) at (-0.8, 0) {};
		\node [style=none] (2) at (0.8, 0) {};
	\end{pgfonlayer}
	\begin{pgfonlayer}{edgelayer}
		\draw [style=noire] (1.center) to (0);
		\draw [style=noire] (0) to (2.center);
	\end{pgfonlayer}
\end{tikzpicture}
$) increases the power of contexts to distinguish purified channels. To characterise the indistinguishability of purified channels with arbitrary contexts, we introduce \emph{second-level} superoperators and  \emph{second-level} transformation matrices: 

\begin{definition}[Second-level Superoperator and  Transformation Matrix]\label{def:S2T2} Given a purified $\H$-channel ${[U,\ket{\varepsilon},\E]}$, let $\S^{(2)}_{[U,\ket{\varepsilon},\E]}: \mathcal L(\H^{\otimes 2}) \to \mathcal L(\H^{\otimes 2}) = \rho \mapsto \textup{Tr}_{\E}\big( U^{(2)} (\rho \otimes \ket{\varepsilon}\!\bra{\varepsilon}) U^{(2)\dagger} \big)$ be the ``second-level'' superoperator and $T^{(2)}_{[U,\ket{\varepsilon},\E]} \coloneqq (I_{\H^{\otimes 2}}\otimes \bra{\varepsilon}) U^{(2)} (I_{\H^{\otimes 2}}\otimes \ket{\varepsilon}) \in \mathcal L(\H^{\otimes 2})$ be the ``second-level'' transformation matrix of $[U,\ket{\varepsilon},\E]$, where $U^{(2)} \coloneqq (I_\H \otimes U) (\mathfrak{S} \otimes I_\E) (I_\H \otimes U)$ and $\mathfrak{S}\coloneqq\ket{\psi_1}\otimes \ket{\psi_2} \mapsto \ket{\psi_2}\otimes \ket{\psi_1}$ is the swap operator.
Graphically, $U^{(2)} = \input{U2.tikz}\,$,
$$\S^{(2)}_{[U,\ket{\varepsilon},\E]} \coloneqq \input{SP2_L-ground.tikz} \qquad \text{and} \qquad T^{(2)}_{[U,\ket{\varepsilon},\E]} \coloneqq \input{TM2_L.tikz}$$
\end{definition}

\begin{restatable}{theorem}{thmcaracobsequivleveltwo}\label{caracobsequivlevel2}
Given two purified $\H$-channels $[U,\ket{\varepsilon},\E]$ and $[U',\ket{\varepsilon'},\E']$, $[U,\ket{\varepsilon},\E] \approx_2 [U',\ket{\varepsilon'},\E']$ iff
they have the same (first level) transformation matrix, the same second level superoperator and the same second level transformation matrix.
 Graphically,
  \begin{subequations}
    \begin{empheq}[left={{[U,\ket{\varepsilon},\E]} \approx_2 {[U',\ket{\varepsilon'},\E']} \quad \text{iff} \ \ \empheqlbrace \ }]{align}
      \input{TM1_L.tikz}\ &= \ \input{TM1_R.tikz} \label{TM1} \tag{T1} \\[3mm]
      \input{SP2_L-ground.tikz}\ &= \ \input{SP2_R-ground.tikz} \label{SP2} \tag{S2} \\[3mm]
      \input{TM2_L.tikz}\ &= \ \input{TM2_R.tikz} \label{TM2} \tag{T2}
    \end{empheq}
  \end{subequations}
\end{restatable}

 The proof is given in Appendix~\ref{appendix_obs_equiv_general} and has the same structure as that of Theorem~\ref{caracobsequivnegfree}.

The contexts used in the proof to show that the constraints~\eqref{SP2} and~\eqref{TM2} are required are of the form $\input{contexteV0V1V0V1variante.tikz}$ and $\input{boucletraversetrouV,1trou_shorter.tikz}$, respectively, for some specific choices of purified channels $[V_0,\ket{\eta_0},\H\otimes\CC^2]$, $[V_1,\ket{\eta_1},\H\otimes\CC^2]$ and $[V,1,\CC]$. Hence, if either the second level superoperators or the second level transformation matrices of two purified channels differ, then the channels can be distinguished by using such contexts.

One may have expected the condition~\eqref{SP1}---i.e., that the two channels have the same first-level superoperator---to also appear in Theorem~\ref{caracobsequivlevel2} (as it did in the previous two cases).
This would however have been redundant, as can be seen from the following remark (also proven in Appendix~\ref{appendix_obs_equiv_general}):

\begin{restatable}{remark}{remSPtwoimpliesSPone}\label{SP2impliesSP1}
Two purified channels $[U,\ket{\varepsilon},\E]$ and $[U',\ket{\varepsilon'},\E']$ having  the same second level  superoperator also have the same first level superoperator, i.e.,
Condition \eqref{SP2} implies \eqref{SP1}. 
\end{restatable}

We note, on the other hand, that the three remaining conditions~\eqref{TM1}, \eqref{SP2} and~\eqref{TM2} are nonredundant. I.e., for each of the three there exist cases where only this condition is not satisfied, and where $[U,\ket{\varepsilon},\E]$ and $[U',\ket{\varepsilon}',\E']$ can be distinguished. E.g., with $\E=\E'=\CC$, $U = I_\H, U' = -I_\H, \ket{\varepsilon} = \ket{\varepsilon'} = 1$, only~\eqref{TM1} fails to hold; with $\H=\E=\E'=\CC^2$, $U = \textsc{Cnot}, U' = (\sqrt{Z} \otimes Z) \textsc{Cnot}, \ket{\varepsilon} = \ket{\varepsilon'} = \ket{0}$, only~\eqref{SP2} fails to hold; and with $\H=\E=\E'=\CC^2$, $U = I_\H \otimes X, U' = I_\H \otimes ZX, \ket{\varepsilon} = \ket{\varepsilon'} = \ket{0}$, only~\eqref{TM2} fails to be satisfied.\footnote{Where $Z=\begin{pmatrix}1&0\\0&-1\end{pmatrix}$, $\sqrt{Z}=\begin{pmatrix}1&0\\0&\mathrm{i}\end{pmatrix}$ and $\textsc{Cnot}=\begin{pmatrix}1&0&0&0\\0&1&0&0\\0&0&0&1\\0&0&1&0\end{pmatrix}$.}

\section{Observational equivalence beyond PBS-diagrams}

In this section, we define a new equivalence relation, inspired by the uniqueness (up to an isometry) of Stinespring's dilations, which subsumes the observational equivalences defined so far.
For that let us first introduce an isometry-based preorder over purified channels:

\begin{definition} Given two purified $\H$-channels $[U,\ket{\varepsilon},\E]$ and $[U'\!\!,\ket{\varepsilon'},\E']$, one has $[U,\ket{\varepsilon},\E] \triangleleft_{iso} [U',\ket{\varepsilon'},\E']$ if there exists  an isometry $W\!\! :\! \E \!\to\! \E'$ s.t. $W\ket{\varepsilon} \!=\! \ket{\varepsilon'}$ and $ (I_\H\otimes W)U\!=\! U' (I_\H\otimes W) $.  
In pictures:
$$\input{epsilonW-HEEprime.tikz}\ =\ \ket{\varepsilon'}\qquad\qquad\qquad\input{UWb-HEEprime.tikz}\ =\ \input{WbUprime-HEEprime.tikz}$$
\end{definition}

Note that $\triangleleft_{iso}$ is not an equivalence relation. It is not symmetric; moreover, its symmetric closure is not transitive.%
\footnote{Taking $\H = \CC$, one has $[1,1,\mathbb C]\triangleleft_{iso} [I_{\CC^2},\ket 0,\mathbb C^2]$ (with $W=\ket 0$) but $\neg ([{I_{\CC^2}},\ket 0,\mathbb C^2]\ \triangleleft_{iso}\ [1,1,\mathbb C])$ (as there is no isometry from $\mathbb C^2$ to $\mathbb C$). With the Pauli operator $Z=\begin{psmallmatrix}1&0\\[1mm]0&-1\end{psmallmatrix}$ one also has $[1,1,\mathbb C]\triangleleft_{iso} [Z,\ket 0,\mathbb C^2]$ (again with $W = \ket 0$), but $[I_{\CC^2},\ket 0,\mathbb C^2]$ and $[Z,\ket 0,\mathbb C^2]$ are not in relation since there is no unitary $W$ such that $W I_{\CC^2} = ZW$ (as $I_{\CC^2}$ and $Z$ have distinct eigenvalues).}
This leads us to consider the following:

\begin{definition}[Iso-equivalence]
The \emph{iso-equivalence} of purified channels is defined as  the symmetric and transitive closure of  $\triangleleft_{iso}$: $\approx_{iso}:= \triangleleft_{iso}^*$. 
\end{definition}

The iso-equivalence is a candidate for characterising indistinguishability of purified channels in more general coherent-control settings. Actually, if $[U,\ket{\varepsilon},\E]$ and $[U',\ket{\varepsilon'},\E']$ are two iso-equivalent purified channels, then intuitively, in any coherent-control setting, $[U,\ket{\varepsilon},\E]$ can be replaced by $[U',\ket{\varepsilon'},\E']$ without changing the global behaviour. Indeed, the evolution of the environment associated with the purified channel is roughly speaking the same (up to the isometry $W$): initialised in the state $W\ket{\varepsilon}$(and with the data register in the state $\ket{\phi}$), the application of $U'$ leads to the state $U'(I_\H\otimes W) (\ket \phi \otimes \ket {\varepsilon})$, which is equal to $(I_\H\otimes W) U( \ket \phi \otimes \ket {\varepsilon})$. So applying $U'$ somehow first cancels the application of $W$, then  applies $U$, and finally applies  $W$ again---which will be cancelled again by the next application of $U'$, and so on. The last application of $W$ is absorbed when the environment is traced out.
In pictures:
{\[\input{iso_eq_1-.tikz}\ =\ \input{iso_eq_2-.tikz}\ =\ \input{iso_eq_3-.tikz}\]\smallskip
\[\hspace{1.5cm}\ =~\ldots ~=\ \input{iso_eq_4-.tikz}\ =\ \input{iso_eq_5-.tikz}\]}

In the framework of PBS-diagrams, one can actually show that the iso-equivalence subsumes, but does not coincide with the $\approx_2$-equivalence (which in turn subsumes the $\approx_1$- and $\approx_0$-equivalences; we provide the proof in Appendix~\ref{appendix:equiv}).

\begin{restatable}{proposition}{propdiff}\label{prop:diff}  $\approx_{iso} ~\subsetneq~ \approx_2  ~\subsetneq~ \approx_1~\subsetneq~ \approx_0 $.
\end{restatable}

Although for PBS-diagrams, the $\approx_2$-equivalence characterises the observational equivalence of purified channels, it could thus be that more general coherent-control settings may distinguish $\approx_2$-equivalent channels. For instance one can imagine including nonpolarising beam splitters, or more general rotations of the polarisation than just the negation, or even settings with ``higher-dimensional polarisations'', which would allow a particle to go more than twice through each gate. Such a setting would be able for instance to distinguish the pair of purified channels used in the proof of Proposition \ref{prop:diff}.
 
We conjecture that two purified channels are not  iso-equivalent if and only if they can be distinguished by some coherently-controlled quantum computation. Here, the notion of coherently-controlled quantum computation is left loosely defined, and   corresponds intuitively to some generalisation of PBS-diagrams allowing a particle to go through a gate an arbitrary number of times.

\section{Discussion}

In this work, we have extended the PBS-diagrams framework of~\cite{clement2020pbs} to allow for the coherent control of more general quantum channels, described as purified channels. 
By defining observational equivalence relations, we have characterised which purified channels  are distinguishable depending on the class of contexts allowed (defined as PBS-diagrams with a hole). 
We also proposed the more refined iso-equivalence, which appears as a candidate for channel indistinguishability in more general coherent-control setups than PBS-diagrams. However, unlike the previous equivalence relations that can be verified with simple criteria---by comparing superoperators and transformation matrices---the iso-equivalence, defined as a transitive closure, is \emph{a priori} not as easy to check in general.

The framework of PBS-diagrams considered here has a number of limitations, which could be lifted in future works. For instance, it would be of practical interest to allow for nonpolarising beam splitters and more general operations on the polarisation; to consider using higher-dimensional control systems, with generalised PBS; or to consider several particles going through the diagrams, possibly correlating the different local environments for future uses of the diagrams, and/or inducing interference effects. We note also that in our description of purified channels, the state of the environment does not evolve by itself, except when the flying particle goes through the channel and the unitary $U$ is applied to the joint system. In fact, as long as each channel is used at most twice (as it was case in this paper), any free evolution of the environment between two uses could be included in $U$; however, introducing such an evolution could make a difference if the channels are used more than twice, and the evolution is different between different uses.

Other open questions raised by our work here include equipping extended PBS-diagrams with an equational theory, as was done in~\cite{clement2020pbs} for the case of ``pure'' PBS-diagrams; lifting our observational equivalences to the diagrams themselves; and investigating more general coherent-control settings, to check in particular whether our iso-equivalence is indeed the good definition for general distinguishability, and if it has a more operational characterisation.

\bibliography{ref}

\appendix

\section{Structural congruence of PBS-diagrams}
\label{appendix_congruence}

Bare \textup{PBS}-diagrams, extended \textup{PBS}-diagrams and contexts are defined up to the congruence generated by the following equalities (with all $n,m,k\geq0$), where $I_n$ is the ``identity diagram'' $I_n := \oplus^n (\begin{tikzpicture}
	\begin{pgfonlayer}{nodelayer}
		\node [style=none] (0) at (-0.3, 0) {};
		\node [style=none] (1) at (0.3, 0) {};
	\end{pgfonlayer}
	\begin{pgfonlayer}{edgelayer}
		\draw [style=noire] (0.center) to (1.center);
	\end{pgfonlayer}
\end{tikzpicture})$ (graphically: $I_n = \input{Inaccoladeg.tikz}$, with $I_0 = \begin{tikzpicture}[scale=0.3]
	\begin{pgfonlayer}{nodelayer}
		\node [style=sdiagrammevide] (0) at (0, 0) {};
	\end{pgfonlayer}
\end{tikzpicture}
$); $\sigma_{1,n}$ is the ``first-wire-goes-last diagram'' defined inductively by $\sigma_{1,0}:=\begin{tikzpicture}
	\begin{pgfonlayer}{nodelayer}
		\node [style=none] (0) at (-0.3, 0) {};
		\node [style=none] (1) at (0.3, 0) {};
	\end{pgfonlayer}
	\begin{pgfonlayer}{edgelayer}
		\draw [style=noire] (0.center) to (1.center);
	\end{pgfonlayer}
\end{tikzpicture}$ and $\sigma_{1,n+1}:=(I_n\oplus\input{swap-xs.tikz})\circ(\sigma_{1,n}\oplus\begin{tikzpicture}
	\begin{pgfonlayer}{nodelayer}
		\node [style=none] (0) at (-0.3, 0) {};
		\node [style=none] (1) at (0.3, 0) {};
	\end{pgfonlayer}
	\begin{pgfonlayer}{edgelayer}
		\draw [style=noire] (0.center) to (1.center);
	\end{pgfonlayer}
\end{tikzpicture})$ (graphically: $\sigma_{1,n} = \input{swap1naccoladeg.tikz}$); and $D:n$ denotes here either a bare \textup{PBS}-diagram $D:n$, an extended \textup{PBS}-diagram $D:\typ{}n$, or a context $C[\cdot]:\typ{}n$:

\begin{itemize}

\item \emph{Neutrality of the identity:} for any $D:n$,
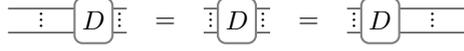
\[\begin{array}{rcccl}D\circ I_n&=&D&=&I_n\circ D\\\\
\input{idDmultifils.tikz}&=&\input{fig_generic_D.tikz}&=&\input{Didmultifils.tikz}
\\&\end{array}\]

\item \emph{Neutrality of the empty diagram:} for any $D:n$, 
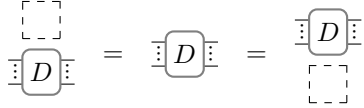
\[\begin{array}{rcccl}\begin{tikzpicture}[scale=0.3]
	\begin{pgfonlayer}{nodelayer}
		\node [style=sdiagrammevide] (0) at (0, 0) {};
	\end{pgfonlayer}
\end{tikzpicture}
\oplus D&=&D&=&D\oplus\begin{tikzpicture}[scale=0.3]
	\begin{pgfonlayer}{nodelayer}
		\node [style=sdiagrammevide] (0) at (0, 0) {};
	\end{pgfonlayer}
\end{tikzpicture}
\\\\
\input{diagrammevidesurD-s.tikz}&=&\input{fig_generic_D.tikz}&=&\input{Dsurdiagrammevide-s.tikz}\\&\end{array}\]

\item \emph{Associativity of the sequential composition:} for any $D_1,D_2,D_3:n$, 
\[\begin{array}{rcl}(D_3\circ D_2)\circ D_1&=&D_3\circ(D_2\circ D_1)\\\\
\input{D1puisD2D3.tikz}&=&\input{D1D2puisD3.tikz}\\&\end{array}\]

\item \emph{Associativity of the parallel composition:} for any $D_1:n,D_2:m\text{ and }D_3:k$, 
\[\begin{array}{rcl}(D_1\oplus D_2)\oplus D_3&=&D_1\oplus(D_2\oplus D_3)\\\\
\input{D1surD2puissurD3.tikz}&=&\input{D1puissurD2surD3.tikz}\\&\end{array}\]

\item \emph{Compatibility of the sequential and parallel compositions:}
for any $D_1,D_2:n$ and $D_3,D_4:m$,
\[\begin{array}{rcl}(D_2\circ D_1)\oplus (D_4\circ D_3)&=&(D_2\oplus D_4)\circ(D_1\oplus D_3)\\\\
\input{D1D2surD3D4.tikz}&=&\input{D1surD3puisD2surD4.tikz}\\&\end{array}\]

\item \emph{Naturality of the swap:}
for any $D:n$,
\[
\begin{array}{rcl}\sigma_{1,n}\circ(\begin{tikzpicture}
	\begin{pgfonlayer}{nodelayer}
		\node [style=none] (0) at (-0.3, 0) {};
		\node [style=none] (1) at (0.3, 0) {};
	\end{pgfonlayer}
	\begin{pgfonlayer}{edgelayer}
		\draw [style=noire] (0.center) to (1.center);
	\end{pgfonlayer}
\end{tikzpicture}\oplus D)&=&(D\oplus \begin{tikzpicture}
	\begin{pgfonlayer}{nodelayer}
		\node [style=none] (0) at (-0.3, 0) {};
		\node [style=none] (1) at (0.3, 0) {};
	\end{pgfonlayer}
	\begin{pgfonlayer}{edgelayer}
		\draw [style=noire] (0.center) to (1.center);
	\end{pgfonlayer}
\end{tikzpicture})\circ\sigma_{1,n}\\\\
\input{Dswap1n.tikz}&=&\input{swap1nD.tikz}\\&\end{array}\]

\item \emph{Inverse law:}
\[
\begin{array}{rcl}\input{swap-xs.tikz}\circ\input{swap-xs.tikz}&=&I_{2}\\\\
\input{swapswap-xxs.tikz}&=&\begin{tikzpicture}[scale=0.4]
	\begin{pgfonlayer}{nodelayer}
		\node [style=none] (0) at (-3, 1) {};
		\node [style=none] (1) at (3, 1) {};
		\node [style=none] (2) at (3, -1) {};
		\node [style=none] (3) at (-3, -1) {};
	\end{pgfonlayer}
	\begin{pgfonlayer}{edgelayer}
		\draw [style=new] (0.center) to (1.center);
		\draw [style=new] (3.center) to (2.center);
	\end{pgfonlayer}
\end{tikzpicture}
\\&\end{array}\]

\item \emph{Naturality in the input:} 
for any $D_1:n$ and $D_2:n+1$,
\[\begin{array}{rcl}Tr(D_2\circ(D_1\oplus\begin{tikzpicture}
	\begin{pgfonlayer}{nodelayer}
		\node [style=none] (0) at (-0.3, 0) {};
		\node [style=none] (1) at (0.3, 0) {};
	\end{pgfonlayer}
	\begin{pgfonlayer}{edgelayer}
		\draw [style=noire] (0.center) to (1.center);
	\end{pgfonlayer}
\end{tikzpicture}))&=&Tr(D_2)\circ D_1\\\\
\input{traceD1hD2.tikz}&=&\input{D1puistraceD2.tikz}\\&\end{array}\]

\item \emph{Naturality in the output:} 
for any $D_1:n+1$ and $D_2:n$,
\[\begin{array}{rcl}Tr((D_2\oplus\begin{tikzpicture}
	\begin{pgfonlayer}{nodelayer}
		\node [style=none] (0) at (-0.3, 0) {};
		\node [style=none] (1) at (0.3, 0) {};
	\end{pgfonlayer}
	\begin{pgfonlayer}{edgelayer}
		\draw [style=noire] (0.center) to (1.center);
	\end{pgfonlayer}
\end{tikzpicture})\circ D_1)&=&D_2\circ Tr(D_1)\\\\
\input{traceD1D2h.tikz}&=&\input{traceD1puisD2.tikz}\\&\end{array}\]

\item \emph{Dinaturality:} for any $D_1:n+m$ and $D_2:m$,
\[\begin{array}{rcl}Tr^m((I_n\oplus D_2)\circ D_1)&=&Tr^m(D_1\circ(I_n\oplus D_2))\\\\
\input{traceD1D2b.tikz}&=&\input{traceD2bD1.tikz}\end{array}\]
where $Tr^m$ denotes the $m^\text{th}$ power of the trace operation.\bigskip

\item \emph{Superposing:} for any $D_1:n$ and $D_2:m+1$,
\[\begin{array}{rcl}Tr(D_1\oplus D_2)&=&D_1\oplus Tr(D_2)\\\\
\input{traceD1surD2.tikz}&=&\input{D1surtraceD2.tikz}\\&\end{array}\]

\item \emph{Yanking:}
\[\begin{array}{rcl}Tr(\input{swap-xs.tikz})&=&\begin{tikzpicture}
	\begin{pgfonlayer}{nodelayer}
		\node [style=none] (0) at (-0.3, 0) {};
		\node [style=none] (1) at (0.3, 0) {};
	\end{pgfonlayer}
	\begin{pgfonlayer}{edgelayer}
		\draw [style=noire] (0.center) to (1.center);
	\end{pgfonlayer}
\end{tikzpicture}\\\\
\input{yankingcentresurfil-s.tikz}&=&\begin{tikzpicture}
	\begin{pgfonlayer}{nodelayer}
		\node [style=none] (0) at (-1.125, 0) {};
		\node [style=none] (1) at (1.125, 0) {};
	\end{pgfonlayer}
	\begin{pgfonlayer}{edgelayer}
		\draw [style=noire] (0.center) to (1.center);
	\end{pgfonlayer}
\end{tikzpicture}
\\&\end{array}\]

\end{itemize}

These equalities are the coherence axioms of a traced PROP, that is, a PROP that is also a traced symmetric monoidal category. An explicit definition of the concept of traced PROP is given in~\cite{clement2020pbs}. See also~\cite{MacLane1965} and~\cite{zanasi2015tel} for a definition of PROPs and further details about them.

\section{Properties of the word path semantics}
\label{sec:appendix:pathsem}

\subsection{Well-definedness and compatibility with the structural congruence}
\label{appendix:semanticsdeform}

It can be proved in the same way as for Propositions~5 and~6 in~\cite{clement2020pbs}, that the word path semantics is well-defined despite the restriction that $k\leq2$ in Rule $(\mathsf T_k)$, that it is deterministic (i.e., that for any bare diagram $D:n$, polarisation $c\in\hv$ and position $p\in[n]$, there exist some unique $c'$, $p'$ and $w$ such that $(D,c,p)\xRightarrow{w}(c',p')$ --- which allows us to define $c^D_{c,p}$, $p^D_{c,p}$ and $w^D_{c,p}$), and that conversely, for any target polarisation $c'$ and position $p'$, there exist $c$ and $p$ such that $(D,c,p)\xRightarrow{w}(c',p')$ for some $w$ (in other words, the map $(c,p)\mapsto(c^{D}_{c,p},p^{D}_{c,p})$ is a bijection). We give here some additional details about the fact that it is invariant modulo diagram deformation:

\begin{proposition}
The word path semantics is invariant modulo diagram deformation.
\end{proposition}

\begin{proof}
One has to check, for each of the equalities given in \cref{appendix_congruence}, that the two sides have the same word path semantics. This is straightforward in each case except for dinaturality. In this case we first prove that Rule $(\mathsf T_k^m)$ below follows from those of \cref{wordpathsemantics}:
\[
\dfrac{D:n+m\qquad \forall i\in\{0,\ldots,k\}, (D,c_i,p_i)\xRightarrow{w_i\,}(c_{i+1},p_{i+1})\qquad(p_{i+1}\geq n){\Leftrightarrow}(i<k)}{(Tr^m(D),c_0,p_0)\xRightarrow{w_0\cdots w_k\,}(c_{k+1},p_{k+1})}(\mathsf T_k^m)
\]
for all $k,m\in\N$.

To prove this, we proceed by induction on $m$. The case $m=0$ is trivial, and the case $m=1$ corresponds to Rule $(\mathsf T_k)$ of \cref{wordpathsemantics} (the rule follows even for $k\geq 3$ since it is then not possible to satisfy its premises).

Now, assume that Rule $(\mathsf T_k^m)$ follows from those of \cref{wordpathsemantics}. Let $D:n+m+1$. Let $c_0\in\hv$ and $p_0\in[n]$. Let $(c_1,p_1),\ldots,(c_{k+1},p_{k+1})$ be the (unique) sequence of couples such that $\forall i\in\{0,\ldots,k\}, (D,c_i,p_i)\xRightarrow{w_i\,}(c_{i+1},p_{i+1})$ and $(p_{i+1}\geq n){\Leftrightarrow}(i<k)$ (that is, $k+1$ is the first index after $0$ such that $p_{k+1}<n$). 
Let $(c_{i_0},p_{i_0}),\ldots,(c_{i_{k'+1}},p_{i_{k'+1}})$, with $0=i_0<i_1<\cdots<i_{k'}<i_{k'+1}=k+1$, be the subsequence of $(c_1,p_1),\ldots,(c_{k+1},p_{k+1})$ where all couples with $p_i=n+m$ have been removed. For each $j\in\{0,\ldots,k'\}$, by Rule $(\mathsf T_k)$ one has $(Tr(D),c_{i_j},p_{i_j})\xRightarrow{w_{i_j}\cdots w_{i_{j+1}-1}}(c_{i_{j+1}},p_{i_{j+1}})$. Additionally, one has $Tr(D):n+m$ and $(p_{i_{j+1}}\geq n)\Leftrightarrow(j<k')$, so that by Rule $(\mathsf T_k^m)$, one has $(Tr^{m+1}(D),c_0,p_0)\xRightarrow{w_0\cdots w_k\,}(c_{k+1},p_{k+1})$, which validates Rule $(\mathsf T_k^{m+1})$.

Given Rule $(\mathsf T_k^m)$ for all $k,m$, we check the compatibility of the word path semantics with dinaturality as follows: given any $D_1:n+m$ and $D_2:m$ with $n,m\geq0$, on the one hand one has
\[\begin{cases}
((I_n\oplus D_2)\circ D_1,c,p)\xRightarrow{w^{D_1}_{c,p}}(c^{D_1}_{c,p},p^{D_1}_{c,p})&\text{if $p^{D_1}_{c,p}<n$}\\
((I_n\oplus D_2)\circ D_1,c,p)\xRightarrow{w^{D_1}_{c,p}w^{D_2}_{\left(c^{D_1}_{c,p}\right),\left(p^{D_1}_{c,p}-n\right)}}(c^{D_2}_{\left(c^{D_1}_{c,p}\right),\left(p^{D_1}_{c,p}-n\right)},p^{D_2}_{\left(c^{D_1}_{c,p}\right),\left(p^{D_1}_{c,p}-n\right)}+n)&\text{if $p^{D_1}_{c,p}\geq n$}
\end{cases}\]
so that given $c_0\in\hv$ and $p_0\in[n]$, if one has a sequence $((I_n\oplus D_2)\circ D_1,c_0,p_0)\xRightarrow{w_0}(c_1,p_1),\ldots,((I_n\oplus D_2)\circ D_1,c_k,p_k)\xRightarrow{w_k}(c_{k+1},p_{k+1})$ with $(p_{i+1}\geq n)\Leftrightarrow(i<k)$, then one has a sequence $(D_1,c_0,p_0)\xRightarrow{w'_0}(c'_1,p'_1),(D_2,c'_1,p'_1-n)\xRightarrow{w''_1}(c_1,p_1-n),(D_1,c_1,p_1)\xRightarrow{w'_1}(c'_1,p'_1),\ldots,(D_1,c_{k-1},p_{k-1})\xRightarrow{w'_{k-1}}(c'_k,p'_k),(D_2,c'_k,p'_k-n)\xRightarrow{w''_{k}}(c_k,p_k-n),(D_1,c_{k},p_{k})\xRightarrow{w'_{k}}(c_{k+1},p_{k+1})$ with $\forall i\in\{0,\ldots,k-1\},w'_iw''_{i+1}=w_i$, and $w'_k=w_k$, so that $(Tr^m((I_n\oplus D_2)\circ D_1),c_0,p_0)\xRightarrow{w'_0w''_1\cdots w'_{k-1}w''_{k}w'_k}(c_{k+1},p_{k+1})$.

On the other hand, one has
\[\begin{cases}
(D_1\circ (I_n\oplus D_2),c,p)\xRightarrow{w^{D_1}_{c,p}}(c^{D_1}_{c,p},p^{D_1}_{c,p})&\text{if $p<n$}\smallskip\\
(D_1\!\circ\! (I_n\!\oplus\! D_2),c,p)\!\xRightarrow{\!\!\!\!w^{D_2}_{c,p-n}w^{D_1}_{\!(c^{D_2}_{c,p-n}),(p^{D_2}_{c,p-n}\!+n)}\!\!\!\!}(c^{D_1}_{\!(c^{D_2}_{c,p-n}),(p^{D_2}_{c,p-n}\!+n)},p^{D_1}_{\!(c^{D_2}_{c,p-n}),(p^{D_2}_{c,p-n}\!+n)})&\text{if $p\geq n$}
\end{cases}\]
so that given $c_0\in\hv$ and $p_0\in[n]$, if one has a sequence $(D_1\circ (I_n\oplus D_2),c_0,p_0)\xRightarrow{\tilde w_0}(c'_1,p'_1),\ldots,(D_1\circ (I_n\oplus D_2),c'_k,p'_k)\xRightarrow{\tilde w_k}(c'_{k+1},p'_{k+1})$ with $(p_{i+1}\geq n)\Leftrightarrow(i<k)$, then one has a sequence $(D_1,c_0,p_0)\xRightarrow{w'_0}(c'_1,p'_1),(D_2,c'_1,p'_1-n)\xRightarrow{w''_1}(c_1,p_1-n),(D_1,c_1,p_1)\xRightarrow{w'_1}(c'_1,p'_1),\ldots,(D_1,c_{k-1},p_{k-1})\xRightarrow{w'_{k-1}}(c'_k,p'_k),(D_2,c'_k,p'_k-n)\xRightarrow{w''_{k}}(c_k,p_k-n),(D_1,c_{k},p_{k})\xRightarrow{w'_{k}}(c_{k+1},p_{k+1})$ with $w'_0=\tilde w_0$ and $\forall i\in\{0,\ldots,k-1\},w''_iw'_i=\tilde w_i$, so that one has $(c'_{k+1},p'_{k+1})=(c_{k+1},p_{k+1})$ and $(Tr^m(D_1\circ (I_n\oplus D_2)),c_0,p_0)\xRightarrow{w'_0w''_1\cdots w'_{k-1}w''_{k}w'_k}(c_{k+1},p_{k+1})$.
This proves that the two sides of the equality have the same semantics.
\end{proof}

\subsection {Proof of Proposition \ref{pasplusdedeuxfois}}
\label{app:pasplusdedeuxfois}

\proppasplusdedeuxfois*

\begin{proof}

We proceed by stuctural induction on $D$.
\begin{itemize}

\item If $D=\begin{tikzpicture}[scale=0.3]
	\begin{pgfonlayer}{nodelayer}
		\node [style=sdiagrammevide] (0) at (0, 0) {};
	\end{pgfonlayer}
\end{tikzpicture}
,\begin{tikzpicture}
	\begin{pgfonlayer}{nodelayer}
		\node [style=none] (0) at (-0.3, 0) {};
		\node [style=none] (1) at (0.3, 0) {};
	\end{pgfonlayer}
	\begin{pgfonlayer}{edgelayer}
		\draw [style=noire] (0.center) to (1.center);
	\end{pgfonlayer}
\end{tikzpicture},\begin{tikzpicture}[scale=0.8]
	\begin{pgfonlayer}{nodelayer}
		\node [style=cartouche] (0) at (0, 0) {\tiny$\neg$};
		\node [style=none] (1) at (-0.8, 0) {};
		\node [style=none] (2) at (0.8, 0) {};
	\end{pgfonlayer}
	\begin{pgfonlayer}{edgelayer}
		\draw [style=noire] (1.center) to (0);
		\draw [style=noire] (0) to (2.center);
	\end{pgfonlayer}
\end{tikzpicture}
,\input{swap-xs.tikz}\text{ or }\input{beamsplitter-xs.tikz}$, then the sums are 0 (they are in particular empty for $D=\begin{tikzpicture}[scale=0.3]
	\begin{pgfonlayer}{nodelayer}
		\node [style=sdiagrammevide] (0) at (0, 0) {};
	\end{pgfonlayer}
\end{tikzpicture}
$), so the result is trivially true.

\item If $D=\begin{tikzpicture}[scale=0.8]
	\begin{pgfonlayer}{nodelayer}
		\node [style=newe] (0) at (0, 0) {\phantom{$a$}};
		\node [style=none] (1) at (-1.2, 0) {};
		\node [style=none] (2) at (1.2, 0) {};
		\node [style=none] (3) at (0.8, -0.35) {\footnotesize$a$};
	\end{pgfonlayer}
	\begin{pgfonlayer}{edgelayer}
		\draw [style=new] (1.center) to (2.center);
	\end{pgfonlayer}
\end{tikzpicture}
$, then one has $w^D_{\rightarrow,0}=w^D_{\uparrow,0}=a$, so the result holds.

\item If $D=D_2\circ D_1$ with $\Gamma_1 \vdash D_1:n, \Gamma_2 \vdash D_2:n, \Gamma_1 \cap \Gamma_2 = \emptyset$, then
\[\begin{disarray}{rcl}\sum_{\begin{scriptarray}{c}\vspace{-1.75em}\\\changelargeursicentre{c\in\{\rightarrow,\uparrow\}}{}\\[-0.4em]p\in[n]\vspace{-0.2em}\end{scriptarray}}|w^D_{c,p}|_{a}&=&\sum_{\begin{scriptarray}{c}\vspace{-1.75em}\\\changelargeursicentre{c\in\{\rightarrow,\uparrow\}}{}\\[-0.4em]p\in[n]\vspace{-0.2em}\end{scriptarray}}\Bigl|w^{D_1}_{c,p}w^{D_2}_{c^{D_1}_{c,p},p^{D_1}_{c,p}}\Bigr|_{a}
\quad = \quad \sum_{\begin{scriptarray}{c}\vspace{-1.75em}\\\changelargeursicentre{c\in\{\rightarrow,\uparrow\}}{}\\[-0.4em]p\in[n]\vspace{-0.2em}\end{scriptarray}}|w^{D_1}_{c,p}|_{a}+\sum_{\begin{scriptarray}{c}\vspace{-1.75em}\\\changelargeursicentre{c\in\{\rightarrow,\uparrow\}}{}\\[-0.4em]p\in[n]\vspace{-0.2em}\end{scriptarray}}\Bigl|w^{D_2}_{c^{D_1}_{c,p},p^{D_1}_{c,p}}\Bigr|_{a}\,.\end{disarray}\]
Since the map $(c,p)\mapsto(c^{D_1}_{c,p},p^{D_1}_{c,p})$ is a bijection, the sum above is equal to
\[\sum_{\begin{scriptarray}{c}\vspace{-1.75em}\\\changelargeursicentre{c\in\{\rightarrow,\uparrow\}}{}\\[-0.4em]p\in[n]\vspace{-0.2em}\end{scriptarray}}|w^{D_1}_{c,p}|_{a}+\sum_{\begin{scriptarray}{c}\vspace{-1.75em}\\\changelargeursicentre{c\in\{\rightarrow,\uparrow\}}{}\\[-0.4em]p\in[n]\vspace{-0.2em}\end{scriptarray}}|w^{D_2}_{c,p}|_{a}\,.\]
Since $D_1$ and $D_2$ have disjoint alphabets $\Gamma_1$ and $\Gamma_2$, at least one of the two sums is equal to 0, and by induction hypothesis, the other one is no greater than $2$.

Moreover, if $D$ is $\begin{tikzpicture}[scale=0.8]
	\begin{pgfonlayer}{nodelayer}
		\node [style=cartouche] (0) at (0, 0) {\tiny$\neg$};
		\node [style=none] (1) at (-0.8, 0) {};
		\node [style=none] (2) at (0.8, 0) {};
	\end{pgfonlayer}
	\begin{pgfonlayer}{edgelayer}
		\draw [style=noire] (1.center) to (0);
		\draw [style=noire] (0) to (2.center);
	\end{pgfonlayer}
\end{tikzpicture}
$-free then for any $c\in\hv$, \[\sum_{p\in[n]}|w^D_{c,p}|_{a}=\sum_{p\in[n]}|w^{D_1}_{c,p}|_{a}+\sum_{p\in[n]}\Bigl|w^{D_2}_{c^{D_1}_{c,p},p^{D_1}_{c,p}}\Bigr|_{a} \,.\]
It is easy to see that since $D_1$ is $\begin{tikzpicture}[scale=0.8]
	\begin{pgfonlayer}{nodelayer}
		\node [style=cartouche] (0) at (0, 0) {\tiny$\neg$};
		\node [style=none] (1) at (-0.8, 0) {};
		\node [style=none] (2) at (0.8, 0) {};
	\end{pgfonlayer}
	\begin{pgfonlayer}{edgelayer}
		\draw [style=noire] (1.center) to (0);
		\draw [style=noire] (0) to (2.center);
	\end{pgfonlayer}
\end{tikzpicture}
$-free, it cannot change the polarisation so that $c^{D_1}_{c,p}=c$. Moreover, the map $(c,p)\mapsto(c,p^{D_1}_{c,p})$ is again a bijection, so that the sum above is equal to
\[\sum_{p\in[n]}|w^{D_1}_{c,p}|_{a}+\sum_{p\in[n]}|w^{D_2}_{c,p}|_{a}\,.\]
Since $D_1$ and $D_2$ have disjoint alphabets, at least one of the two sums is equal to 0, and by induction hypothesis, the other one is no greater than $1$.

\item If $D=D_1\oplus D_2$ with $\Gamma_1 \vdash D_1:n_1, \Gamma_2 \vdash D_2:n_2$ such that $n_1 + n_2 = n, \Gamma_1 \cap \Gamma_2 = \emptyset$, then
\[\sum_{\begin{scriptarray}{c}\vspace{-1.75em}\\\changelargeursicentre{c\in\{\rightarrow,\uparrow\}}{}\\[-0.4em]p\in[n]\vspace{-0.2em}\end{scriptarray}}|w^D_{c,p}|_{a}\ \ = \ \ \sum_{\begin{scriptarray}{c}\vspace{-1.75em}\\\changelargeursicentre{c\in\{\rightarrow,\uparrow\}}{}\\[-0.4em]p\in[n_1]\vspace{-0.2em}\end{scriptarray}}|w^D_{c,p}|_{a}+\sum_{\begin{scriptarray}{c}\vspace{-1.75em}\\\changelargeursicentre{c\in\{\rightarrow,\uparrow\}}{}\\[-0.4em]\changelargeursicentre{n_1\leq p<n}{p\in[n_1]}\vspace{-0.2em}\end{scriptarray}}|w^D_{c,p}|_{a}
\ \ = \ \ \sum_{\begin{scriptarray}{c}\vspace{-1.75em}\\\changelargeursicentre{c\in\{\rightarrow,\uparrow\}}{}\\[-0.4em]p\in[n_1]\vspace{-0.2em}\end{scriptarray}}|w^{D_1}_{c,p}|_{a}+\sum_{\begin{scriptarray}{c}\vspace{-1.75em}\\\changelargeursicentre{c\in\{\rightarrow,\uparrow\}}{}\\[-0.4em]p\in[n_2]\vspace{-0.2em}\end{scriptarray}}|w^{D_2}_{c,p}|_{a}\,.\]
Since $D_1$ and $D_2$ have disjoint alphabets $\Gamma_1$ and $\Gamma_2$, at least one of the two sums is equal to 0, and by induction hypothesis, the other one is no greater than $2$.

Moreover, if $D$ is $\begin{tikzpicture}[scale=0.8]
	\begin{pgfonlayer}{nodelayer}
		\node [style=cartouche] (0) at (0, 0) {\tiny$\neg$};
		\node [style=none] (1) at (-0.8, 0) {};
		\node [style=none] (2) at (0.8, 0) {};
	\end{pgfonlayer}
	\begin{pgfonlayer}{edgelayer}
		\draw [style=noire] (1.center) to (0);
		\draw [style=noire] (0) to (2.center);
	\end{pgfonlayer}
\end{tikzpicture}
$-free then similarly, for any $c\in\hv$,
\[\begin{disarray}{rcl}\sum_{p\in[n]}|w^D_{c,p}|_{a}=\sum_{p\in[n_1]}|w^{D_1}_{c,p}|_{a}+\sum_{p\in[n_2]}|w^{D_2}_{c,p}|_{a}\,.\end{disarray}\]
Since $D_1$ and $D_2$ have disjoint alphabets, at least one of the two sums is equal to 0, and by induction hypothesis, the other one is no greater than $1$.

\item If $D=Tr(D')$ with $D':n+1$, then for any $c\in\hv$ and any $p \in [n]$,%
\footnote{The argument that follows applies to $n \ge 1$; for $n=0$ the sums are again empty (as in the case of $D=\begin{tikzpicture}[scale=0.3]
	\begin{pgfonlayer}{nodelayer}
		\node [style=sdiagrammevide] (0) at (0, 0) {};
	\end{pgfonlayer}
\end{tikzpicture}
$), so that the result trivially holds.}
the couple $(c^D_{c,p},p^D_{c,p})$ is the unique couple such that there exists a sequence of arrows $(D',c,p)\xRightarrow{w_0\,}(c_1,n)$, $(D',c_1,n)\xRightarrow{w_1\,}(c_2,n),\ldots,(D',c_{k-1},n)\xRightarrow{w_{k-1}\,}(c_k,n)$, $(D',c_k,n)\xRightarrow{w_k\,}(c^D_{c,p},p^D_{c,p})$ (we additionally know that $k\leq2$). Given such a sequence, one has $|w^D_{c,p}|_{a}=|w^{D'}_{c,p}|_{a}+|w^{D'}_{c_1,n}|_{a}+\cdots+|w^{D'}_{c_k,n}|_{a}$.

Since the map $(c',p')\mapsto(c^{D'}_{c',p'},p^{D'}_{c',p'})$ is a bijection, a given couple $(c',p')$, now with $p'\in[n+1]$, cannot appear more than once on the left of an arrow (i.e., as a polarisation and position configuration entering the diagram $D'$) among the family of all possible such sequences. In particular for $p'=n$, it follows that the sum of all partial sums $|w^{D'}_{c_1,n}|_{a}+\cdots+|w^{D'}_{c_k,n}|_{a}$ above, for all possible sequences (i.e., for all starting configurations $c,p$), is upper-bounded by $\sum_{c\in\hv}|w^{D'}_{c,n}|_{a}$.
 Therefore,
\[\sum_{\begin{scriptarray}{c}\vspace{-1.75em}\\\changelargeursicentre{c\in\{\rightarrow,\uparrow\}}{}\\[-0.4em]p\in[n]\vspace{-0.2em}\end{scriptarray}}|w^D_{c,p}|_{a}\leq\sum_{\begin{scriptarray}{c}\vspace{-1.75em}\\\changelargeursicentre{c\in\{\rightarrow,\uparrow\}}{}\\[-0.4em]p\in[n]\vspace{-0.2em}\end{scriptarray}}|w^{D'}_{c,p}|_{a} + \sum_{c\in\hv}|w^{D'}_{c,n}|_{a} = \sum_{\begin{scriptarray}{c}\vspace{-1.75em}\\\changelargeursicentre{c\in\{\rightarrow,\uparrow\}}{}\\[-0.4em]p\in[n+1]\vspace{-0.2em}\end{scriptarray}}|w^{D'}_{c,p}|_{a}\]
which, by induction hypothesis, is no greater than $2$.

Moreover, if $D$ is $\begin{tikzpicture}[scale=0.8]
	\begin{pgfonlayer}{nodelayer}
		\node [style=cartouche] (0) at (0, 0) {\tiny$\neg$};
		\node [style=none] (1) at (-0.8, 0) {};
		\node [style=none] (2) at (0.8, 0) {};
	\end{pgfonlayer}
	\begin{pgfonlayer}{edgelayer}
		\draw [style=noire] (1.center) to (0);
		\draw [style=noire] (0) to (2.center);
	\end{pgfonlayer}
\end{tikzpicture}
$-free then since the polarisation cannot change, one can proceed in the same way for each of the two polarisations $\rightarrow$ and $\uparrow$ separately. We similarly get that for any $c\in\hv$,
\[\sum_{p\in[n]}|w^D_{c,p}|_{a}\leq\sum_{p\in[n+1]}|w^{D'}_{c,p}|_{a}\]
which, by induction hypothesis, is no greater than $1$.
\end{itemize}
\end{proof}

\subsection{Converse: proof of Proposition \ref{recip2fois}}
\label{app:recip2fois}

\proprecipdeuxfois*

\begin{proof}

We prove by induction on  $\sum_{c,p} |w_{c,p}|$ (where $|w|$ denotes the length of the word $w$) that there exists $D$ such that $(D,c,p)\xRightarrow{w^D_{c,p}} (c,p)$, which ensures the proposition.

We say that such a diagram realises the family $W=\{w_{c,p}\}_{(c,p)\in \{\rightarrow,\uparrow\} \times [n]}$.

\begin{itemize}

\item If $\sum_{c,p} |w_{c,p}|=0$, the ``identity'' diagram $I_n = \oplus^n(-)$ gives $(I_n,c,p)\xRightarrow{} (c,p)$, so that $I_n$ realises the family $W=\{w_{c,p} = \epsilon\}_{(c,p)\in \{\rightarrow,\uparrow\} \times [n]}$ (the only one satisfying $\sum_{c,p} |w_{c,p}|=0$). 

\item If  $W=\{w_{c,p}\}_{(c,p)\in \{\rightarrow,\uparrow\} \times [n]}$ is such that $w_{c_0,p_0}=a$ for some $(c_0,p_0)$ and some label $a$, and $w_{c,p}$ is the empty word otherwise (i.e., if $\sum_{c,p} |w_{c,p}|=1$),
then  the  following diagrams realise $W$ when $c_0=\uparrow$ and $c_0=\rightarrow$, respectively: 
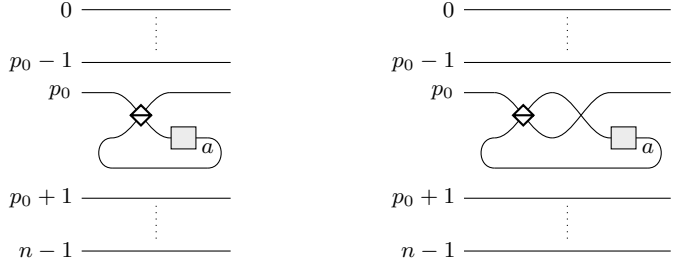
\[\input{prop4cas0vertical.tikz}\qquad\qquad\qquad\input{prop4cas0horizontal.tikz}\] 

\item For any family $W=\{w_{c,p}\}_{(c,p)\in \{\rightarrow,\uparrow\} \times [n]}$ with at least one nonempty word (i.e., with $\sum_{c,p} |w_{c,p}| \ge 1$)
such that every letter appears at most twice in the whole family, consider a nonempty $w_{c_0,p_0}$. It can be written in the form $ua$ with $|a|=1$:

\begin{itemize}

\item  If $\sum_{c,p} |w_{c,p}|_a=1$, then composing a   diagram $D'$ realising $W'= \{w'_{c,p}\}_{(c,p)\in \{\rightarrow,\uparrow\} \times [n]}$  where  $w'_{c_0,p_0}=u$ and $w'_{c,p}= w_{c,p}$ otherwise (which exists by induction) with a diagram  $D_a$ realising (as in the previous case) $W''= \{w''_{c,p}\}_{(c,p)\in \{\rightarrow,\uparrow\} \times [n]}$ such that $w''_{c_0,p_0}=a$  and  $w''_{c,p}$ is the empty word otherwise  allows one to realise $W$.

\item  Otherwise,  there exists a second occurrence of $a$ in some $w_{c_1,p_1}$, that one can write in the form $w_{c_1,p_1}=vaw$ with $a\notin v$.

\begin{itemize}

\item If $p_1=p_0$ and $c_0=c_1$ then $\exists\tilde{w},\ w_{c_0,p_0}=va\tilde{w}a$.  Let $D'$ be a diagram on $n+1$ wires realising $w'_{c_0,p_0}=v$, $w'_{c_0,n}=\tilde{w}$,
$w'_{\neg c_0,n}=\epsilon$ (where $\neg(\rightarrow)=\uparrow, \neg(\uparrow)=\rightarrow$) and $w'_{ c,p}=w_{c,p}$ on the first $n$ wires otherwise. The following diagrams realise $W$ when $c_0=\uparrow$ and $c_0=\rightarrow$, respectively: 
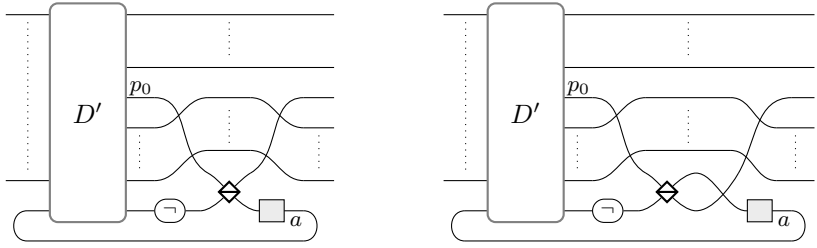
\[\input{prop4cas2vertical.tikz}\qquad\qquad\input{prop4cas2horizontal.tikz}\]

\item  If $p_1=p_0$ and $c_0\neq c_1$ then $w_{c_0,p_0}=ua$ and $w_{c_1,p_0}=vaw$.
  Let $D'$ be a diagram on $n+1$ wires realising $w'_{c_0,p_0}=u$, $w'_{c_1,p_0}=v$, $w'_{c_0,n}=\epsilon$,
  $w'_{c_1,n}=w$ and $w'_{ c,p}=w_{c,p}$ on the first $n$ wires otherwise. The following diagram realises $W$:
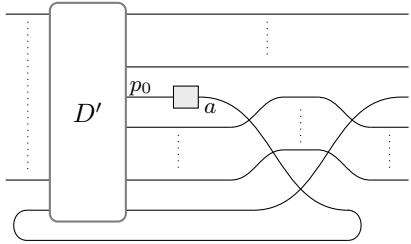
\[\input{prop4cas1bis.tikz}\]

\item  If $p_1\neq p_0$ and $c_0=c_1$ then $w_{c_0,p_0}=ua$ and $w_{c_0,p_1}=vaw$. Let $D'$ be a diagram on $n+1$ wires realising $w'_{c_0,p_0}=u$, $w'_{c_0,n}=\epsilon$, $w'_{c_0,p_1}=v$, $w'_{\neg c_0,n}=w$, and $w'_{c,p}=w_{c,p}$ on the first $n$ wires otherwise. The following diagram realises $W$ when $c_0=\uparrow$:
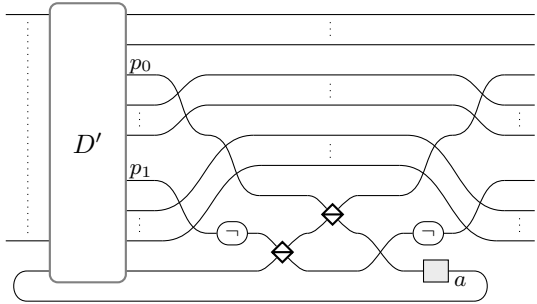
\[\input{prop4cas3vertical.tikz}\]
and the following diagram realises $W$ when $c_0=\rightarrow$:
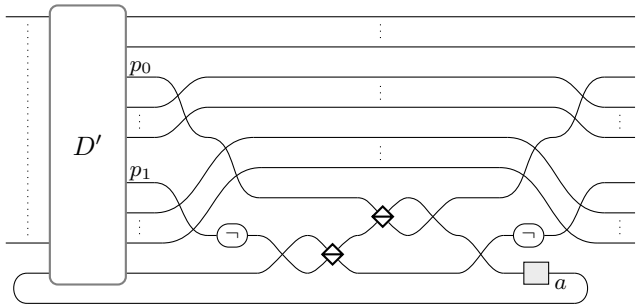
\[\input{prop4cas3horizontal.tikz}\]

\item  If $p_1\neq p_0$ and $c_0\neq c_1$,    let $D'$ be a diagram on $n+1$ wires realising $w'_{c_0,p_0}=u$, $w'_{c_0,n}=\epsilon$, $w'_{c_1,p_1}=v$, $w'_{c_1,n}=w$, and $w'_{ c,p}=w_{c,p}$ on the first $n$ wires otherwise. The following diagram realises $W$ with  $c_0=\uparrow$:
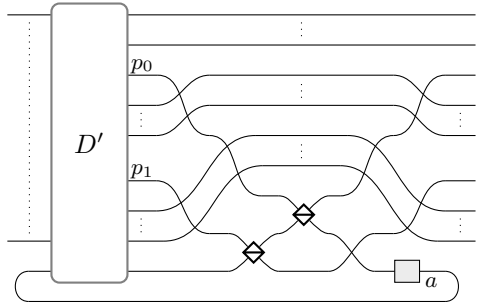
\[\input{prop4cas4vertical.tikz}\]
and the following diagram realises $W$ with  $c_0=\rightarrow$:
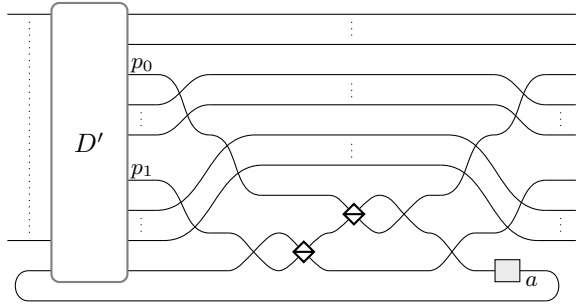
\[\input{prop4cas4horizontal.tikz}\]

\end{itemize}
\end{itemize}
\end{itemize}

Note that for the cases where $p_0\neq p_1$, although strictly speaking the last four pictures illustrate the case where $p_0<p_1$, they aim at representing the general case. If $p_1<p_0$, then one should include a swap between the two corresponding wires in order to connect them to the appropriate ports.

Note that this proof is constructive, although not deterministic. That is, by following the induction steps, one can build a diagram realising a given family $W$; although, depending on how one follows these steps (i.e., on which word $w_{c_0,p_0}$ one singles out at each step), one may end up with different possible diagrams. 
Moreover, the only cases where some $\begin{tikzpicture}[scale=0.8]
	\begin{pgfonlayer}{nodelayer}
		\node [style=cartouche] (0) at (0, 0) {\tiny$\neg$};
		\node [style=none] (1) at (-0.8, 0) {};
		\node [style=none] (2) at (0.8, 0) {};
	\end{pgfonlayer}
	\begin{pgfonlayer}{edgelayer}
		\draw [style=noire] (1.center) to (0);
		\draw [style=noire] (0) to (2.center);
	\end{pgfonlayer}
\end{tikzpicture}
$ are added are the cases where the letter $a$ under consideration appears twice for the same polarisation $c_0$. Therefore, if every letter appears at most once for each polarisation $c$, then any diagram built by unfolding the induction is $\begin{tikzpicture}[scale=0.8]
	\begin{pgfonlayer}{nodelayer}
		\node [style=cartouche] (0) at (0, 0) {\tiny$\neg$};
		\node [style=none] (1) at (-0.8, 0) {};
		\node [style=none] (2) at (0.8, 0) {};
	\end{pgfonlayer}
	\begin{pgfonlayer}{edgelayer}
		\draw [style=noire] (1.center) to (0);
		\draw [style=noire] (0) to (2.center);
	\end{pgfonlayer}
\end{tikzpicture}
$-free. This proves the second statement.

\end{proof}

\section{Circuit notations}
\label{appendix_circuit_notations}

In this paper, we further develop the graphical representation of coherent control by means of \textup{PBS}-diagrams, but we also use circuit-like notations when it is convenient to represent sequential and parallel compositions of linear transformations 
$\Hin \to  \Hout$ for some Hilbert spaces $\Hin$ and $\Hout$ (e.g., unitary operations, density matrices or matrices of the form $\ket i \bra j$) and linear maps $\Lin(\Hin )\to \Lin(\Hout)$ (i.e., superoperators).
We briefly review these circuit-like notations: given a linear transformation $U : \H_1\otimes \ldots \otimes \H_n \to  \H'_1\otimes \ldots \otimes \H'_k $, $$\input{circ_Unk.tikz}$$ is a circuit of type $ \H_1\otimes \ldots \otimes \H_n \to  \H'_1\otimes \ldots \otimes \H'_k 
$.
 Note that the Hilbert spaces on the wires are generally omitted when these are clear from the context. 

The identity operator on a Hilbert space is represented as a wire. Sequential composition consists in plugging two circuits (with the appropriate types) in a row, and tensor product consists in putting two circuits in parallel, e.g., for any linear maps $U: \H_0\to \H_1$, $V: \H_1\to \H_2$, $W:\H_2\to \H_3$:\\
 \centerline{$\input{circ_UVT2.tikz}~=~\input{circ_UVT1.tikz} \qquad\qquad \input{circ_UtensVT2L.tikz}~=~\input{circ_UtensVT2R.tikz}$}

The associativity of both $\circ$ and $\otimes$, and the mixed-product property ($(U'\otimes V')\circ (U\otimes V) = (U'\circ U)\otimes (V'\circ V)$ for some $U: \H_0\to \H_1$, $U': \H_1\to \H_2$, $V: \H_3\to \H_4$, $V': \H_4\to \H_5$) guarantee the nonambiguity of the circuit-like notations. 
Quantum states (resp. their adjoints) can be added to input (resp. output) wires, e.g., $\input{circ_Uinout.tikz} = \bra \psi U \ket \varphi$. 

We extend these notations to represent partial trace:%
\footnote{Given a linear transformation $C\colon\mathcal A\otimes\mathcal B\to\mathcal A'\otimes\mathcal B$, its partial trace over $\mathcal B$ is defined as $\textup{Tr}_{\mathcal B}(C)=\sum_i (I_{\mathcal A'}\otimes\bra{i}_{\mathcal B}) C (I_{\mathcal A}\otimes\ket{i}_{\mathcal B})$ where $\{\ket{i}_{\mathcal B}\}_i$ is an orthonormal basis of $\mathcal B$.}
 given a linear transformation $C:\H_1\otimes \ldots \otimes \H_n\otimes\H \to  \H'_1\otimes \ldots \otimes \H'_k\otimes\H$:
\[\input{circ_CnktraceH.tikz}\ =\ \input{circ_TrCnk.tikz}\]

With this trace and the swap $\input{swapH1H2.tikz}=\ket{\varphi_1}\otimes\ket{\varphi_2}\mapsto\ket{\varphi_2}\otimes\ket{\varphi_1}$, and with quantum states $\ket\varphi\in\H$ (resp. their adjoints $\bra\psi\in\H^\dag$) seen as linear transformations $\CC\to\H$ (resp. $\H\to\CC$), circuits form a traced strict symmetric monoidal category.
That is, in addition to the fact that the notation is not ambiguous, circuits can be deformed at will (as long as their topology is preserved) without changing the transformation that is represented.

Following~\cite{CP10-env,CJPV-ICALP19}, we further extend these notations to represent linear maps $\Lin(\Hin )\to \Lin(\Hout)$, using the ``ground'' symbol $\input{ground.tikz}$. Given a ``pure'' (i.e., $\input{ground.tikz}~$-free) circuit, plugging one (or several) $\input{ground.tikz}$ in its output wire(s) corresponds essentially to tracing out the corresponding systems---or more precisely, to defining the map that takes an operator (typically, a density matrix, $\rho$) acting on the input Hilbert spaces, applies the linear map defined by the circuit (as in $\rho \mapsto U\rho U^\dagger$), and traces out the systems to which the ground symbol is attached, e.g., $$
\begin{array}{rcccl}
\input{circ_trace0.tikz}~ &=&~ \rho \mapsto \textup{Tr}_{\H_2'\otimes \H_3'}(U\rho  U^\dagger) &=& \rho \mapsto \input{circ_trace0-developpe.tikz}\\&&\\
 \input{circ_trace1.tikz}~ &=&~ \rho \mapsto \textup{Tr}_{\E}(V(\rho\otimes \ket{\varphi}\bra{\varphi}) V^\dagger) &=& \rho \mapsto \input{circ_trace1-developpe.tikz}\\
 \end{array}$$
where the top example defines a map $\Lin(\H_1\otimes\H_2\otimes\H_3)\to\Lin(\H_1')$, and the bottom example defines a map $\Lin(\H_0)\to\Lin(\H_1)$.
We say that such circuits are of type $\Lin(\Hin )\to \Lin(\Hout)$.

\begin{remark} With these definitions, for a circuit with input Hilbert spaces $\H_1, \ldots, \H_n$ and output Hilbert spaces $\H'_1, \ldots, \H'_k$ to represent a linear map $\Lin( \H_1\otimes \ldots \otimes \H_n)\to \Lin(   \H'_1\otimes \ldots \otimes \H'_k )$, it  must contain at least one $\input{ground.tikz}~$ symbol. 
As a consequence the CPTP map $\rho \mapsto U\rho U^\dagger$ cannot be represented as $\begin{tikzpicture}
	\begin{pgfonlayer}{nodelayer}
		\node [style=none] (1) at (-0.75, 0) {};
		\node [style=none] (9) at (0.75, 0) {};
		\node [style=boite2] (12) at (0, 0) {$U$};
	\end{pgfonlayer}
	\begin{pgfonlayer}{edgelayer}
		\draw (1.center) to (9.center);
	\end{pgfonlayer}
\end{tikzpicture}
$ (which is a ``pure'' circuit) but for instance as $\input{circ_pure_ground.tikz}$. 
\end{remark}

Note that one can consider $\input{groundH.tikz}$ as a generator $\L(\H)\to\L(\CC)=\CC$ and place it anywhere in the circuit. Because of the traced strict symmetric monoidal structure of $\input{ground.tikz}~$-free circuits and the fact that $\input{swapgroundbas.tikz}=\input{groundsurfil.tikz}$, this does not create ambiguity since all ways of pulling the $\input{ground.tikz}$ symbols to the right give the same linear map. Moreover, circuits with this additional generator still form a traced strict symmetric monoidal category.

\begin{remark}
When quantum states are attached to all input wires of a circuit, the circuit represents a linear map $\CC\to\Lin(\H)$, of the form $\lambda\mapsto\lambda\rho$ for some mixed state $\rho\in\Lin(\H)$. By a slight abuse of notation, we identify this linear map with the state $\rho$ itself.
\end{remark}

\section{Observational equivalence of purified channels}
\label{appendix_obs_equiv}

\subsection{Using PBS-free contexts: proof of Theorem \ref{thm:cptpchannels}} \label{appendix:PBSfree}

\thmcptpchannels*

\bigskip

\begin{proof}
By considering the trivial context $\begin{tikzpicture}[scale=0.8]
	\begin{pgfonlayer}{nodelayer}
		\node [style=newe] (0) at (0, 0) {$\cdot$};
		\node [style=none] (1) at (-0.95, 0) {};
		\node [style=none] (2) at (0.95, 0) {};
	\end{pgfonlayer}
	\begin{pgfonlayer}{edgelayer}
		\draw [style=new] (1.center) to (2.center);
	\end{pgfonlayer}
\end{tikzpicture}
$, if $[U,\ket{\varepsilon},\E] \approx_0 [U',\ket{\varepsilon'},\E']$ then in particular, $\interp{\begin{tikzpicture}[scale=0.8]
	\begin{pgfonlayer}{nodelayer}
		\node [style=newe] (0) at (0, 0) {\footnotesize $U,\ket{\varepsilon}$};
		\node [style=none] (1) at (-1.6, 0) {};
		\node [style=none] (2) at (1.6, 0) {};
	\end{pgfonlayer}
	\begin{pgfonlayer}{edgelayer}
		\draw [style=new] (1.center) to (2.center);
	\end{pgfonlayer}
\end{tikzpicture}}=\interp{\begin{tikzpicture}[scale=0.8]
	\begin{pgfonlayer}{nodelayer}
		\node [style=newe] (0) at (0, 0) {\footnotesize $U',\ket{\varepsilon'}$};
		\node [style=none] (1) at (-1.85, 0) {};
		\node [style=none] (2) at (1.85, 0) {};
	\end{pgfonlayer}
	\begin{pgfonlayer}{edgelayer}
		\draw [style=new] (1.center) to (2.center);
	\end{pgfonlayer}
\end{tikzpicture}
}$, hence, $\S^{(1)}_{[U,\ket{\varepsilon},\E]}=\S^{(1)}_{[U',\ket{\varepsilon'},\E']}$.

Conversely, let us assume that $\S^{(1)}_{[U,\ket{\varepsilon},\E]}=\S^{(1)}_{[U',\ket{\varepsilon'},\E']}$. Let $C[\cdot]\in\C_0$. By deformation of diagrams one can write it in one of the following two forms:
\begin{itemize}
\item $C'[\cdot]\oplus D$, with $D:\H^{(0)}$ and $C'[\cdot]$ of the form 
\[\input{contexteC0general.tikz}\]
for some purified channels $[V_i,\ket{\eta_i},\mathcal V_i], [W_j,\ket{\zeta_j},\mathcal Z_j] \in \mathfrak{C}({\mathcal H})$, and where $\begin{tikzpicture}[scale=0.6]
	\begin{pgfonlayer}{nodelayer}
		\node [style=none] (1) at (-1.325, 0) {};
		\node [style=none] (2) at (1.325, 0) {};
		\node [style=negserie, inner sep=0.25em] (4) at (0, 0) {};
	\end{pgfonlayer}
	\begin{pgfonlayer}{edgelayer}
		\draw [style=noire] (4) to (2.center);
		\draw [style=noire] (1.center) to (4);
	\end{pgfonlayer}
\end{tikzpicture}
$ denotes any sequence of $\begin{tikzpicture}[scale=0.8]
	\begin{pgfonlayer}{nodelayer}
		\node [style=cartouche] (0) at (0, 0) {\tiny$\neg$};
		\node [style=none] (1) at (-0.8, 0) {};
		\node [style=none] (2) at (0.8, 0) {};
	\end{pgfonlayer}
	\begin{pgfonlayer}{edgelayer}
		\draw [style=noire] (1.center) to (0);
		\draw [style=noire] (0) to (2.center);
	\end{pgfonlayer}
\end{tikzpicture}
$, possibly of length $0$;
\item $D\oplus C'[\cdot]$, with $D:\H^{(1)}$ and $C'[\cdot]:\H^{(0)}$.
\end{itemize}

In the latter case, the semantics does not depend on what is plugged in the hole, so that $\interp{C[U,\ket{\varepsilon},\E]}=\interp{C[U',\ket{\varepsilon'},\E']}$. In the former case, 
\[\begin{array}{r@{\,}c@{\,}l}\interp{C[U,\ket{\varepsilon},\E]}&=&
\mathcal{X}^* \otimes {\mathcal I}_\CC \otimes \left(\S_{[W_\ell,\ket{\zeta_\ell},\mathcal Z_\ell]}^{(1)}\circ\cdots\circ\S_{[W_1,\ket{\zeta_1},\mathcal Z_1]}^{(1)}\circ \S_{[U,\ket \varepsilon,\E]}^{(1)}\circ \S_{[V_k,\ket{\eta_k},\mathcal V_k]}^{(1)}\circ\cdots\circ\S_{[V_1,\ket{\eta_1},\mathcal V_1]}^{(1)}\right)\smallskip\\
&=&
\mathcal{X}^* \otimes {\mathcal I}_\CC \otimes \left(\S_{[W_\ell,\ket{\zeta_\ell},\mathcal Z_\ell]}^{(1)}\circ\cdots\circ\S_{[W_1,\ket{\zeta_1},\mathcal Z_1]}^{(1)}\circ \S_{[U',\ket \varepsilon',\E']}^{(1)}\circ \S_{[V_k,\ket{\eta_k},\mathcal V_k]}^{(1)}\circ\cdots\circ\S_{[V_1,\ket{\eta_1},\mathcal V_1]}^{(1)}\right)\smallskip\\
&=&\interp{C[U',\ket{\varepsilon'},\E']}\end{array}\] 
where $\mathcal{X}^*$ is either the identity map over $\L(\CC^{\hv})$ if the total number of $\begin{tikzpicture}[scale=0.8]
	\begin{pgfonlayer}{nodelayer}
		\node [style=cartouche] (0) at (0, 0) {\tiny$\neg$};
		\node [style=none] (1) at (-0.8, 0) {};
		\node [style=none] (2) at (0.8, 0) {};
	\end{pgfonlayer}
	\begin{pgfonlayer}{edgelayer}
		\draw [style=noire] (1.center) to (0);
		\draw [style=noire] (0) to (2.center);
	\end{pgfonlayer}
\end{tikzpicture}
$ in $C'[\cdot]$ is even, or the linear map $\ket{c}\bra{c'}\mapsto\ket{\neg {c}}\bra{\neg {c'}}$ if the total number of $\begin{tikzpicture}[scale=0.8]
	\begin{pgfonlayer}{nodelayer}
		\node [style=cartouche] (0) at (0, 0) {\tiny$\neg$};
		\node [style=none] (1) at (-0.8, 0) {};
		\node [style=none] (2) at (0.8, 0) {};
	\end{pgfonlayer}
	\begin{pgfonlayer}{edgelayer}
		\draw [style=noire] (1.center) to (0);
		\draw [style=noire] (0) to (2.center);
	\end{pgfonlayer}
\end{tikzpicture}
$ in $C'[\cdot]$ is odd, and $\mathcal I_\CC$ is the identity map over $\CC$. Hence, $[U,\ket{\varepsilon},\E] \approx_0 [U',\ket{\varepsilon'},\E']$.
\end{proof}

\subsection{Using negation-free contexts: proof of Theorem~\ref{caracobsequivnegfree}}
\label{appendix_obs_equiv_negfree}

\thmcaracobsequivnegfree*
\bigskip

We are going to prove at the same time \cref{caracobsequivnegfree} and the fact that allowing multiple input/output wires does not increase the power of $\begin{tikzpicture}[scale=0.8]
	\begin{pgfonlayer}{nodelayer}
		\node [style=cartouche] (0) at (0, 0) {\tiny$\neg$};
		\node [style=none] (1) at (-0.8, 0) {};
		\node [style=none] (2) at (0.8, 0) {};
	\end{pgfonlayer}
	\begin{pgfonlayer}{edgelayer}
		\draw [style=noire] (1.center) to (0);
		\draw [style=noire] (0) to (2.center);
	\end{pgfonlayer}
\end{tikzpicture}
$-free contexts, stated as the following proposition:
\begin{proposition}\label{multiwirecontextsnegfree}
Given two purified $\H$-channels $[U,\ket{\varepsilon},\E]$ and $[U',\ket{\varepsilon'},\E']$, one has $[U,\ket{\varepsilon},\E]\approx_1[U',\ket{\varepsilon'},\E']$ (that is, for any $\begin{tikzpicture}[scale=0.8]
	\begin{pgfonlayer}{nodelayer}
		\node [style=cartouche] (0) at (0, 0) {\tiny$\neg$};
		\node [style=none] (1) at (-0.8, 0) {};
		\node [style=none] (2) at (0.8, 0) {};
	\end{pgfonlayer}
	\begin{pgfonlayer}{edgelayer}
		\draw [style=noire] (1.center) to (0);
		\draw [style=noire] (0) to (2.center);
	\end{pgfonlayer}
\end{tikzpicture}
$-free context $C[\cdot]:\H^{(1)}$, $\interp{C[U,\ket{\varepsilon},\E]}=\interp{C[U',\ket{\varepsilon'},\E']}$) if and only if for any $\begin{tikzpicture}[scale=0.8]
	\begin{pgfonlayer}{nodelayer}
		\node [style=cartouche] (0) at (0, 0) {\tiny$\neg$};
		\node [style=none] (1) at (-0.8, 0) {};
		\node [style=none] (2) at (0.8, 0) {};
	\end{pgfonlayer}
	\begin{pgfonlayer}{edgelayer}
		\draw [style=noire] (1.center) to (0);
		\draw [style=noire] (0) to (2.center);
	\end{pgfonlayer}
\end{tikzpicture}
$-free context $C[\cdot]:\H^{(n)}$, $\interp{C[U,\ket{\varepsilon},\E]}=\interp{C[U',\ket{\varepsilon'},\E']}$.
\end{proposition}

Namely, what we are going to prove is the following lemma:

\begin{lemma}\label{caracobsequivnegfreestrong}
Given two purified $\H$-channels $[U,\ket{\varepsilon},\E]$ and $[U',\ket{\varepsilon'},\E']$, the following three statements are equivalent:
\begin{enumerate}[(I)]
\item\label{contextequiv1} $[U,\ket{\varepsilon},\E]\approx_1[U',\ket{\varepsilon'},\E']$, that is, for any $\begin{tikzpicture}[scale=0.8]
	\begin{pgfonlayer}{nodelayer}
		\node [style=cartouche] (0) at (0, 0) {\tiny$\neg$};
		\node [style=none] (1) at (-0.8, 0) {};
		\node [style=none] (2) at (0.8, 0) {};
	\end{pgfonlayer}
	\begin{pgfonlayer}{edgelayer}
		\draw [style=noire] (1.center) to (0);
		\draw [style=noire] (0) to (2.center);
	\end{pgfonlayer}
\end{tikzpicture}
$-free context $C[\cdot]:\H^{(1)}$, $\interp{C[U,\ket{\varepsilon},\E]}=\interp{C[U',\ket{\varepsilon'},\E']}$\bigskip
\item\label{contextequivn} for any $\begin{tikzpicture}[scale=0.8]
	\begin{pgfonlayer}{nodelayer}
		\node [style=cartouche] (0) at (0, 0) {\tiny$\neg$};
		\node [style=none] (1) at (-0.8, 0) {};
		\node [style=none] (2) at (0.8, 0) {};
	\end{pgfonlayer}
	\begin{pgfonlayer}{edgelayer}
		\draw [style=noire] (1.center) to (0);
		\draw [style=noire] (0) to (2.center);
	\end{pgfonlayer}
\end{tikzpicture}
$-free context $C[\cdot]:\H^{(n)}$, $\interp{C[U,\ket{\varepsilon},\E]}=\interp{C[U',\ket{\varepsilon'},\E']}$\bigskip
\item\label{sameS1andT1} $\S^{(1)}_{[U,\ket{\varepsilon},\E]}=\S^{(1)}_{[U',\ket{\varepsilon'},\E']}$ and $T^{(1)}_{[U,\ket{\varepsilon},\E]}=T^{(1)}_{[U',\ket{\varepsilon'},\E']}$
\end{enumerate}
\end{lemma}
\bigskip

\begin{proof}[Proof of \cref{caracobsequivnegfree,multiwirecontextsnegfree}]
\Cref{caracobsequivnegfree} is exactly $\eqref{contextequiv1}\Leftrightarrow\eqref{sameS1andT1}$, while \cref{multiwirecontextsnegfree} is $\eqref{contextequiv1}\Leftrightarrow\eqref{contextequivn}$.
\end{proof}

\begin{proof}[Proof of \cref{caracobsequivnegfreestrong}]
It is clear that $\eqref{contextequivn}\Rightarrow\eqref{contextequiv1}$. Therefore, what one has to prove is that $\eqref{sameS1andT1}\Rightarrow\eqref{contextequivn}$ (that is, the conditions given by \cref{caracobsequivnegfree} are sufficient even with contexts with mutiple input/output wires) and that $\eqref{contextequiv1}\Rightarrow\eqref{sameS1andT1}$ (or equivalently $\neg\eqref{sameS1andT1}\Rightarrow\neg\eqref{contextequiv1}$, that is, these conditions are necessary).

\paragraph*{Proof of strong sufficiency ($\eqref{sameS1andT1}\Rightarrow\eqref{contextequivn}$)}

Let us assume \eqref{sameS1andT1}. Let $C[\cdot]:\H^{(n)}$ be any $\begin{tikzpicture}[scale=0.8]
	\begin{pgfonlayer}{nodelayer}
		\node [style=cartouche] (0) at (0, 0) {\tiny$\neg$};
		\node [style=none] (1) at (-0.8, 0) {};
		\node [style=none] (2) at (0.8, 0) {};
	\end{pgfonlayer}
	\begin{pgfonlayer}{edgelayer}
		\draw [style=noire] (1.center) to (0);
		\draw [style=noire] (0) to (2.center);
	\end{pgfonlayer}
\end{tikzpicture}
$-free context. Let $\Gamma\vdash D:n$ be an underlying bare diagram of both $C[U,\ket{\varepsilon},\E]$ and $C[U',\ket{\varepsilon'},\E']$. Let $\G=([U_x,\ket{\varepsilon_x},\E_x])_{x\in\Gamma}$ and $\G'=([U'_x,\ket{\varepsilon'_x},\E'_x])_{x\in\Gamma}$ be such that $[U_a,\ket{\varepsilon_a},\E_a]=[U,\ket{\varepsilon},\E]$ and $[U'_a,\ket{\varepsilon'_a},\E'_a]=[U',\ket{\varepsilon'},\E']$ for some $a\in\Gamma$, while $[U_x,\ket{\varepsilon_x},\E_x]=[U'_x,\ket{\varepsilon'_x},\E'_x]$ for all $x\in\Gamma\backslash\{a\}$; and let $\F=([U_x,\ket{\varepsilon_x},\E_x])_{x\in\Gamma\backslash\{a\}}$.

Let $c,c'\in\hv$ and $p,p'\in[n]$. By Proposition \ref{pasplusdedeuxfois} one has $|w^D_{c,p}|_{a}\leq1$ and $|w^D_{c'\!,p'}|_{a}\leq1$, so that there are four cases:
\begin{itemize}
\item If $|w^D_{c,p}|_{a}=|w^D_{c'\!,p'}|_{a}=1$, then one can write $w^D_{c,p}=uav$ and $w^D_{c'\!,p'}=u'av'$ with $u,v,u',v'\in(\Gamma\backslash\{a\})^*$. Then with the shorthand notation $\ket{c,p} = \ket{c}\otimes\ket{p}$, for any $\rho\in\mathcal L(\H)$: 
\[
\interp{C[U,\ket{\varepsilon},\E]}(\ket{c,p}\!\bra{c'\!,p'}\otimes \rho) = \ket{c_{c,p}^D,p_{c,p}^D}\!\bra{c_{c'\!,p'}^D,p_{c'\!,p'}^D} \otimes \textup{Tr}_{\mathcal E_\G} \big( V^\G_{w^D_{c,p}} (\rho\otimes \ket{\varepsilon_\G}\!\bra{\varepsilon_\G}){V^\G_{w^D_{c'\!,p'}}}^{\changelargeur{\hspace{-1.2em}\dagger}{}} \big)
\]
with (using the circuit notations defined in Appendix~\ref{appendix_circuit_notations}, and noting for instance that $V^\G_u = V^\F_u \otimes I_\E$ and that $V^\G_a =U \otimes I_{\E_\F}$)
\begin{align*}
& \textup{Tr}_{\mathcal E_\G} ( V^\G_{w^D_{c,p}} (\rho\otimes \ket{\varepsilon_\G}\!\bra{\varepsilon_\G}){V^\G_{w^D_{c'\!,p'}}}^{\changelargeur{\hspace{-1.2em}\dagger}{}}) = \textup{Tr}_{\E_\F,\E} \big( V^\G_v V^\G_a V^\G_u (\rho\otimes \ket{\varepsilon_\F}\!\bra{\varepsilon_\F}\otimes \ket{\varepsilon}\!\bra{\varepsilon}){V^\G_{u'}}^\dagger {V^\G_a}^\dagger {V^\G_{v'}}^\dagger \big) \\[2mm]
 & \qquad = 
 \input{S1v4depart.tikz} 
 \\[2mm]
 & \qquad = 
 \input{S1v4arrivee.tikz} 
 \\[2mm]
 & \qquad = \textup{Tr}_{\mathcal E_\F} \left( V_v^\F \ \textup{Tr}_{\mathcal E}\big((U\otimes I_{\E_\F})(\sigma_{u,u'}\otimes \ket{\varepsilon}\!\bra{\varepsilon})(U^\dagger\otimes I_{\E_\F} \big) \, {V_{v'}^\F}^\dag \right)
 \\[2mm]
 & \qquad = \textup{Tr}_{\mathcal E_\F} \left( V_v^\F \big(\S^{(1)}_{[U,\ket{\varepsilon},\E]}\otimes{\cal I}_{\E_\F}\big)[\sigma_{u,u'}] {V_{v'}^\F}^\dag \right),
\end{align*}
where ${\cal I}_{\E_\F}$ is the identity map over $\L(\E_\F)$ and $\sigma_{u,u'} = \input{HECuprimerhou.tikz}$. 

Similarly,
\[\interp{C[U',\ket{\varepsilon'},\E']}(\ket{c,p}\!\bra{c'\!,p'}\otimes \rho) = \ket{c_{c,p}^D,p_{c,p}^D}\!\bra{c_{c'\!,p'}^D,p_{c'\!,p'}^D} \otimes\textup{Tr}_{\mathcal E_\F} \left( V_v^\F \big({\cal I}_{\E_\F}\otimes\S^{(1)}_{[U',\ket{\varepsilon'},\E']}\big)[\sigma_{u,u'}] {V_{v'}^\F}^\dag \right).\]
Since $\S^{(1)}_{[U,\ket \varepsilon,\E]}=\S^{(1)}_{[U',\ket{\varepsilon'},\E']}$, this is equal to $\interp{C[U,\ket{\varepsilon},\E]}(\ket{c,p}\!\bra{c'\!,p'}\otimes \rho)$.

\item If $|w^D_{c,p}|_{a}=1$ and $|w^D_{c'\!,p'}|_{a}=0$, then one can write $w^D_{c,p}=uav$ with $u,v\in(\Gamma\backslash\{a\})^*$. Then for any $\rho\in\mathcal L(\H)$:
\[
\interp{C[U,\ket{\varepsilon},\E]}(\ket{c,p}\!\bra{c'\!,p'}\otimes \rho) = \ket{c_{c,p}^D,p_{c,p}^D}\!\bra{c_{c'\!,p'}^D,p_{c'\!,p'}^D} \otimes \textup{Tr}_{\mathcal E_\G} \big( V^\G_{w^D_{c,p}} (\rho\otimes \ket{\varepsilon_\G}\!\bra{\varepsilon_\G}){V^\G_{w^D_{c'\!,p'}}}^{\changelargeur{\hspace{-1.2em}\dagger}{}} \big)
\]
with
\begin{align*}
& \textup{Tr}_{\mathcal E_\G} ( V^\G_{w^D_{c,p}} (\rho\otimes \ket{\varepsilon_\G}\!\bra{\varepsilon_\G}){V^\G_{w^D_{c'\!,p'}}}^{\changelargeur{\hspace{-1.2em}\dagger}{}}) = \textup{Tr}_{\E_\F,\E} \big( V^\G_v V^\G_a V^\G_u (\rho\otimes \ket{\varepsilon_\F}\!\bra{\varepsilon_\F}\otimes \ket{\varepsilon}\!\bra{\varepsilon}){V^\G_{w^D_{c'\!,p'}}}^{\changelargeur{\hspace{-1.2em}\dagger}{}} \big) \\[2mm]
 & \qquad = 
 \input{TM1v4depart.tikz} 
 \\[2mm]
 & \qquad = 
 \input{TM1v4arrivee.tikz} 
 \\[2mm]
 & \qquad = \textup{Tr}_{\mathcal E_\F} \left( V_v^\F \, (I_{\E_\F} \otimes T^{(1)}_{[U,\ket{\varepsilon},\E]}) \, \sigma_{u,c'\!,p'} \right),
 \end{align*}
 where $\sigma_{u,c'\!,p'} = \input{HECwcprimepprimerhou.tikz}$. 

Again, similarly, one has
\[\interp{C[U',\ket{\varepsilon'},\E']}(\ket{c,p}\!\bra{c'\!,p'}\otimes \rho) = \ket{c_{c,p}^D,p_{c,p}^D}\!\bra{c_{c'\!,p'}^D,p_{c'\!,p'}^D} \otimes\textup{Tr}_{\mathcal E_\F} \left( V_v^\F \, (I_{\E_\F} \otimes T^{(1)}_{[U',\ket{\varepsilon'},\E']}) \, \sigma_{u,c'\!,p'} \right).\]
Since $T^{(1)}_{[U,\ket \varepsilon,\E]}=T^{(1)}_{[U',\ket{\varepsilon'},\E']}$, this is equal to $\interp{C[U,\ket{\varepsilon},\E]}(\ket{c,p}\!\bra{c'\!,p'}\otimes \rho)$.

\item The case $|w^D_{c,p}|_{a}=0$ and $|w^D_{c'\!,p'}|_{a}=1$ is similar to the previous case.

\item If $|w^D_{c,p}|_{a}=|w^D_{c'\!,p'}|_{a}=0$, then for any $\rho\in\L(\H)$:
\begin{longtable}{RCL}
\interp{C[U,\ket{\varepsilon},\E]}(\ket{c,p}\!\bra{c'\!,p'}\otimes \rho) &=& \ket{c_{c,p}^D,p_{c,p}^D}\!\bra{c_{c'\!,p'}^D,p_{c'\!,p'}^D} \otimes \textup{Tr}_{\mathcal E_\G} \big( V^\G_{w^D_{c,p}} (\rho\otimes \ket{\varepsilon_\G}\!\bra{\varepsilon_\G}){V^\G_{w^D_{c'\!,p'}}}^{\changelargeur{\hspace{-1.2em}\dagger}{}} \big)\smallskip\\
&=&\ket{c_{c,p}^D,p_{c,p}^D}\!\bra{c_{c'\!,p'}^D,p_{c'\!,p'}^D} \otimes \textup{Tr}_{\mathcal E_\F} \big( V^\F_{w^D_{c,p}} (\rho\otimes \ket{\varepsilon_\F}\!\bra{\varepsilon_\F}){V^\F_{w^D_{c'\!,p'}}}^{\changelargeur{\hspace{-1.2em}\dagger}{}} \big)\smallskip\\
&=&\ket{c_{c,p}^D,p_{c,p}^D}\!\bra{c_{c'\!,p'}^D,p_{c'\!,p'}^D} \otimes \textup{Tr}_{\mathcal E_{\G'}} \big( V^{\G'}_{w^D_{c,p}} (\rho\otimes \ket{\varepsilon_{\G'}}\!\bra{\varepsilon_{\G'}}){V^{\G'}_{w^D_{c'\!,p'}}}^{\changelargeur{\hspace{-1.2em}\dagger}{}} \big)\smallskip\\
&=&\interp{C[U',\ket{\varepsilon'},\E']}(\ket{c,p}\!\bra{c'\!,p'}\otimes \rho).
\end{longtable}
\end{itemize}
We have thus proved that $\interp{C[U,\ket{\varepsilon},\E]}(\ket{c,p}\!\bra{c'\!,p'}\otimes \rho)=\interp{C[U',\ket{\varepsilon'},\E']}(\ket{c,p}\!\bra{c'\!,p'}\otimes \rho)$ for all $c,p,c'\!,p'\text{ and }\rho$, that is, $[U,\ket{\varepsilon},\E]\approx_1[U',\ket{\varepsilon'},\E']$.

\paragraph*{Proof of necessity ($\neg\eqref{sameS1andT1}\Rightarrow\neg\eqref{contextequiv1}$)}

\begin{itemize} 
\item If $\S^{(1)}_{[U,\ket \varepsilon,\E]} \neq \S^{(1)}_{[U',\ket{\varepsilon'},\E']}$, then 
already with the trivial context $\begin{tikzpicture}[scale=0.8]
	\begin{pgfonlayer}{nodelayer}
		\node [style=newe] (0) at (0, 0) {$\cdot$};
		\node [style=none] (1) at (-0.95, 0) {};
		\node [style=none] (2) at (0.95, 0) {};
	\end{pgfonlayer}
	\begin{pgfonlayer}{edgelayer}
		\draw [style=new] (1.center) to (2.center);
	\end{pgfonlayer}
\end{tikzpicture}
$ one can distinguish $[U,\ket{\varepsilon},\E]$ and $[U',\ket{\varepsilon'},\E']$. Indeed, one has $\interp{\begin{tikzpicture}[scale=0.8]
	\begin{pgfonlayer}{nodelayer}
		\node [style=newe] (0) at (0, 0) {\footnotesize $U,\ket{\varepsilon}$};
		\node [style=none] (1) at (-1.6, 0) {};
		\node [style=none] (2) at (1.6, 0) {};
	\end{pgfonlayer}
	\begin{pgfonlayer}{edgelayer}
		\draw [style=new] (1.center) to (2.center);
	\end{pgfonlayer}
\end{tikzpicture}}={\cal I}_{\CC^{\hv}\otimes\CC}\otimes\S^{(1)}_{[U,\ket \varepsilon,\E]}$, whereas $\interp{\begin{tikzpicture}[scale=0.8]
	\begin{pgfonlayer}{nodelayer}
		\node [style=newe] (0) at (0, 0) {\footnotesize $U',\ket{\varepsilon'}$};
		\node [style=none] (1) at (-1.85, 0) {};
		\node [style=none] (2) at (1.85, 0) {};
	\end{pgfonlayer}
	\begin{pgfonlayer}{edgelayer}
		\draw [style=new] (1.center) to (2.center);
	\end{pgfonlayer}
\end{tikzpicture}
}={\cal I}_{\CC^{\hv}\otimes\CC}\otimes\S^{(1)}_{[U',\ket{\varepsilon'},\E']}$ (where ${\cal I}_{\CC^{\hv}\otimes\CC}$ is the identity map over $\L(\CC^{\hv}\otimes\CC)$).

\item If $T^{(1)}_{[U,\ket{\varepsilon},\E]} \neq T^{(1)}_{[U',\ket{\varepsilon'},\E']}$, then by considering the following context:
\[C[\cdot]=\input{boucletraversetrou.tikz}\]
one gets in particular
\[\interp{C[U,\ket{\varepsilon},\E]}(\ket{\uparrow,0}\!\bra{\rightarrow,0}\otimes I_{\H})=\ket{\uparrow,0}\!\bra{\rightarrow,0}\otimes\textup{Tr}_\E\big(U(I_{\H}\otimes\ket{\varepsilon}\!\bra{\varepsilon})\big)=\ket{\uparrow,0}\!\bra{\rightarrow,0}\otimes T^{(1)}_{[U,\ket{\varepsilon},\E]}\]
and similarly
\[\interp{C[U',\ket{\varepsilon'},\E']}(\ket{\uparrow,0}\!\bra{\rightarrow,0}\otimes I_{\H})
=\ket{\uparrow,0}\!\bra{\rightarrow,0}\otimes T^{(1)}_{[U',\ket{\varepsilon'},\E']}.\]

Since $T^{(1)}_{[U,\ket{\varepsilon},\E]}\neq T^{(1)}_{[U',\ket{\varepsilon'},\E']}$, this implies that $[U,\ket{\varepsilon},\E]\not\approx_1[U',\ket{\varepsilon'},\E']$.
\end{itemize}

\end{proof}

\subsection{Using general contexts: proof of Theorem~\ref{caracobsequivlevel2} and Remark~\ref{SP2impliesSP1}}
\label{appendix_obs_equiv_general}

\thmcaracobsequivleveltwo*

\remSPtwoimpliesSPone*
\begin{proof}[Proof of \cref{SP2impliesSP1}]~

\begin{longtable}{RRCL}
\eqref{SP2}\ \Leftrightarrow&\input{SP2_L-ground.tikz}&=&\input{SP2_R-ground.tikz}\\\\
\Rightarrow&\input{SP2_L-doubleground.tikz}&=&\input{SP2_R-doubleground.tikz}\\\\
\Leftrightarrow&\input{SP2_L-groundsousground.tikz}&=&\input{SP2_R-groundsousground.tikz}\\\\
\Leftrightarrow&\input{FU_L.tikz}&=&\input{FU_R.tikz} \qquad \Leftrightarrow \qquad \eqref{SP1} \\\\
\end{longtable}

\end{proof}

\subsubsection*{Proof of \cref{caracobsequivlevel2}}\label{sec:preuvecaracobsequivlevel2}

As we did for \cref{caracobsequivnegfree}, we are going to prove \cref{caracobsequivlevel2} at the same time as the fact that allowing multiple input/output wires in the contexts does not increase their power, stated as the following proposition:
\begin{proposition}\label{multiwirecontexts}
Given two purified $\H$-channels $[U,\ket{\varepsilon},\E]$ and $[U',\ket{\varepsilon'},\E']$, one has $[U,\ket{\varepsilon},\E]\approx_2[U',\ket{\varepsilon'},\E']$ (that is, for any context $C[\cdot]:\H^{(1)}$, $\interp{C[U,\ket{\varepsilon},\E]}=\interp{C[U',\ket{\varepsilon'},\E']}$) if and only if for any context $C[\cdot]:\H^{(n)}$, $\interp{C[U,\ket{\varepsilon},\E]}=\interp{C[U',\ket{\varepsilon'},\E']}$.
\end{proposition}

Namely, what we are going to prove is the following lemma:
\begin{lemma}\label{caracobsequivlevel2strong}
Given two purified $\H$-channels $[U,\ket{\varepsilon},\E]$ and $[U',\ket{\varepsilon'},\E']$, the following three statements are equivalent:
\begin{enumerate}[(I)]
\item\label{contextequiv1level2} $[U,\ket{\varepsilon},\E]\approx_2[U',\ket{\varepsilon'},\E']$, that is, for any context $C[\cdot]:\H^{(1)}$, $\interp{C[U,\ket{\varepsilon},\E]}=\interp{C[U',\ket{\varepsilon'},\E']}$\bigskip
\item\label{contextequivnlevel2} for any context $C[\cdot]:\H^{(n)}$, $\interp{C[U,\ket{\varepsilon},\E]}=\interp{C[U',\ket{\varepsilon'},\E']}$\bigskip
\item\label{sameT1S2T2} $T^{(1)}_{[U,\ket{\varepsilon},\E]}=T^{(1)}_{[U',\ket{\varepsilon'},\E']}$, $\S^{(2)}_{[U,\ket{\varepsilon},\E]}=\S^{(2)}_{[U',\ket{\varepsilon'},\E']}$ and $T^{(2)}_{[U,\ket{\varepsilon},\E]}=T^{(2)}_{[U',\ket{\varepsilon'},\E']}$
\end{enumerate}
\end{lemma}

\begin{proof}[Proof of \cref{caracobsequivlevel2,multiwirecontexts}]
\Cref{caracobsequivlevel2} is exactly $\eqref{contextequiv1level2}\Leftrightarrow\eqref{sameT1S2T2}$, while \cref{multiwirecontexts} is $\eqref{contextequiv1level2}\Leftrightarrow\eqref{contextequivnlevel2}$.
\end{proof}

\begin{proof}[Proof of \cref{caracobsequivlevel2strong}]
The structure of the proof is the same as for \cref{caracobsequivnegfree}. It is clear that $\eqref{contextequivnlevel2}\Rightarrow\eqref{contextequiv1level2}$. Therefore, what one has to prove is that $\eqref{sameT1S2T2}\Rightarrow\eqref{contextequivnlevel2}$ (that is, Conditions \eqref{TM1}, \eqref{SP2} and \eqref{TM2} are sufficient even with contexts with mutiple input/output wires) and that $\eqref{contextequiv1level2}\Rightarrow\eqref{sameT1S2T2}$ (or equivalently $\neg\eqref{sameT1S2T2}\Rightarrow\neg\eqref{contextequiv1level2}$, that is, the three conditions are necessary).

\paragraph*{Proof of strong sufficiency ($\eqref{sameT1S2T2}\Rightarrow\eqref{contextequivnlevel2}$)}

Let us assume \eqref{sameT1S2T2}. Let $C[\cdot]:\H^{(n)}$ be any context. Let $\Gamma\vdash D:n$ be an underlying bare diagram of both $C[U,\ket{\varepsilon},\E]$ and $C[U',\ket{\varepsilon'},\E']$. Let $\G=([U_x,\ket{\varepsilon_x},\E_x])_{x\in\Gamma}$ and $\G'=([U'_x,\ket{\varepsilon'_x},\E'_x])_{x\in\Gamma}$ be such that $[U_a,\ket{\varepsilon_a},\E_a]=[U,\ket{\varepsilon},\E]$ and $[U'_a,\ket{\varepsilon'_a},\E'_a]=[U',\ket{\varepsilon'},\E']$ for some $a\in\Gamma$, while $[U_x,\ket{\varepsilon_x},\E_x]=[U'_x,\ket{\varepsilon'_x},\E'_x]$ for all $x\in\Gamma\backslash\{a\}$; and let $\F=([U_x,\ket{\varepsilon_x},\E_x])_{x\in\Gamma\backslash\{a\}}$.

Let $c,c'\in\hv$ and $p,p'\in[n]$. 
By Proposition \ref{pasplusdedeuxfois}, the possible cases are the following:
\begin{itemize}
\item $|w^D_{c,p}|_{a}\leq 1$ and $|w^D_{c'\!,p'}|_{a}\leq 1$
\item $(c,p)\neq (c'\!,p')$, $|w^D_{c,p}|_{a}=2$ and $|w^D_{c'\!,p'}|_{a}=0$
\item $(c,p)\neq (c'\!,p')$, $|w^D_{c,p}|_{a}=0$ and $|w^D_{c'\!,p'}|_{a}=2$
\item $(c,p)=(c'\!,p')$ and $|w^D_{c,p}|_{a}=2$.
\end{itemize}
The first case can be treated exactly in the same way as in the proof of \cref{caracobsequivnegfreestrong}. To address the other three cases, one can first note that because of the strict symmetric monoidal structure of circuits, for any $V\in\L(\H\otimes\E_\F)$, \[\input{HEEFUVU.tikz}\ =\ \input{HEEFUUVtrace.tikz}.\]
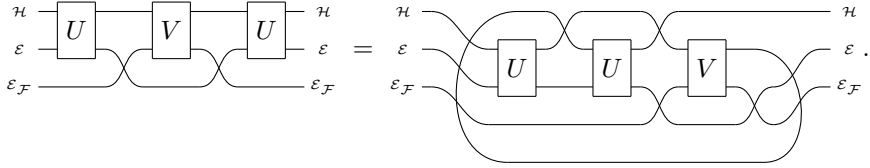
\begin{itemize}
\item If $(c,p)\neq (c'\!,p')$, $|w^D_{c,p}|_{a}=2$ and $|w^D_{c'\!,p'}|_{a}=0$, then one can write $w^D_{c,p}=uavat$ with $u,v,t\in(\Gamma\backslash\{a\})^*$. Then for any $\rho\in\mathcal L(\H)$:
\[
\interp{C[U,\ket{\varepsilon},\E]}(\ket{c,p}\!\bra{c'\!,p'}\otimes \rho) = \ket{c_{c,p}^D,p_{c,p}^D}\!\bra{c_{c'\!,p'}^D,p_{c'\!,p'}^D} \otimes \textup{Tr}_{\mathcal E_\G} \big( V^\G_{w^D_{c,p}} (\rho\otimes \ket{\varepsilon_\G}\!\bra{\varepsilon_\G}){V^\G_{w^D_{c'\!,p'}}}^{\changelargeur{\hspace{-1.2em}\dagger}{}} \big)
\]
with
\begin{longtable}{L}
\textup{Tr}_{\mathcal E_\G} ( V^\G_{w^D_{c,p}} (\rho\otimes \ket{\varepsilon_\G}\!\bra{\varepsilon_\G}){V^\G_{w^D_{c'\!,p'}}}^{\changelargeur{\hspace{-1.2em}\dagger}{}}) = \textup{Tr}_{\E_\F,\E} \big( V^\G_t V^\G_a V^\G_v V^\G_a V^\G_u (\rho\otimes \ket{\varepsilon_\F}\!\bra{\varepsilon_\F}\otimes \ket{\varepsilon}\!\bra{\varepsilon}){V^\G_{w^D_{c'\!,p'}}}^{\changelargeur{\hspace{-1.2em}\dagger}{}} \big) 
\\\\
= \input{TM2v4depart.tikz} 
\\\\
= \input{TM2v4milieu.tikz} 
\\\\
= \input{TM2v4arrivee.tikz} 
\\\\
= \textup{Tr}_{\mathcal E_\F,\H} \left( \sigma_{v,t} \, (T^{(2)}_{[U,\ket{\varepsilon},\E]} \otimes I_{\E_\F}) \, \sigma_{u,c'\!,p'} \right), 
\end{longtable}
 where $\sigma_{v,t}=\input{HHEFvt.tikz}$\bigskip\\ and $\sigma_{u,c'\!,p'}=\input{HEFHwcprimepprimerhou.tikz}$ (and $\textup{Tr}_{\E_\F,\E}\coloneqq\textup{Tr}_{\E_\F}\circ\textup{Tr}_{\E}$, $\textup{Tr}_{\E_\F,\H}\coloneqq\textup{Tr}_{\E_\F}\circ\textup{Tr}_{\H}$; by convention we always take the partial trace over the last factor of the tensor product in both the input and output spaces, so that there is no ambiguity about which copy of $\H$ is traced out in the last formula).\bigskip
 
 Similarly,
 \[
\interp{C[U',\ket{\varepsilon'},\E']}(\ket{c,p}\!\bra{c'\!,p'}\otimes \rho) = \ket{c_{c,p}^D,p_{c,p}^D}\!\bra{c_{c'\!,p'}^D,p_{c'\!,p'}^D} \otimes \textup{Tr}_{\mathcal E_\F,\H} \left( \sigma_{v,t} \, (T^{(2)}_{[U',\ket{\varepsilon'},\E']} \otimes I_{\E_\F}) \, \sigma_{u,c'\!,p'} \right).
\]
Since $T^{(2)}_{[U,\ket{\varepsilon},\E]}=T^{(2)}_{[U',\ket{\varepsilon'},\E']}$, this is equal to $\interp{C[U,\ket{\varepsilon},\E]}(\ket{c,p}\!\bra{c'\!,p'}\otimes \rho)$.
\item The case $(c,p)\neq (c'\!,p')$, $|w^D_{c,p}|_{a}=0$ and $|w^D_{c'\!,p'}|_{a}=2$ is similar to the previous case.
\item If $(c,p)=(c'\!,p')$ and $|w^D_{c,p}|_{a}=2$, then one can again write $w^D_{c,p}(=w^D_{c',p'})=uavat$ with $u,v,t\in(\Gamma\backslash\{a\})^*$. Then for any $\rho\in\mathcal L(\H)$:
\[
\interp{C[U,\ket{\varepsilon},\E]}(\ket{c,p}\!\bra{c,p}\otimes \rho) = \ket{c_{c,p}^D,p_{c,p}^D}\!\bra{c_{c,p}^D,p_{c,p}^D} \otimes \textup{Tr}_{\mathcal E_\G} \big( V^\G_{w^D_{c,p}} (\rho\otimes \ket{\varepsilon_\G}\!\bra{\varepsilon_\G}){V^\G_{w^D_{c,p}}}^{\changelargeur{\hspace{-1.2em}\negphantom{w^D_{c,p}}\phantom{w^D_{c'\!,p'}}\dagger}{}} \big)
\]
with
\begin{longtable}{L}
\textup{Tr}_{\mathcal E_\G} ( V^\G_{w^D_{c,p}} (\rho\otimes \ket{\varepsilon_\G}\!\bra{\varepsilon_\G}){V^\G_{w^D_{c,p}}}^{\changelargeur{\hspace{-1.2em}\negphantom{w^D_{c,p}}\phantom{w^D_{c'\!,p'}}\dagger}{}}) = \textup{Tr}_{\E_\F,\E} \big( V^\G_t V^\G_a V^\G_v V^\G_a V^\G_u (\rho\otimes \ket{\varepsilon_\F}\!\bra{\varepsilon_\F}\otimes \ket{\varepsilon}\!\bra{\varepsilon}) {V^\G_{u}}^\dagger {V^\G_a}^\dagger {V^\G_{v}}^\dagger {V^\G_a}^\dagger {V^\G_{t}}^\dagger \big) \\\\
= \input{S2v4depart.tikz}
 \\\\
= \input{S2v4milieusigma.tikz}
 \\\\
= \input{S2v4arriveesigma.tikz}
 \\\\
\qquad = \textup{Tr}_{\mathcal E_\F,\H,\H} \!\Big(\! (\sigma_{v,t}\!\otimes\! I_\H) \textup{Tr}_\E\! \big[ (U^{(2)}\!\!\otimes\! I_{\E_\F\otimes\H}) (\sigma'_{u,u}\!\otimes\!\ket{\varepsilon}\!\bra{\varepsilon}) (U^{(2)}\!\!\otimes\! I_{\E_\F\otimes\H})^\dagger \big] (\sigma_{v,t}\!\otimes\! I_\H)^\dag(I_{\H\otimes\E_\F}\!\otimes\!\mathfrak S
)\! \Big) \\\\
\qquad = \textup{Tr}_{\mathcal E_\F,\H,\H} \left( (\sigma_{v,t}\otimes I_\H) \, (\S^{(2)}_{[U,\ket{\varepsilon},\E]} \otimes 
{\cal I}_{\E_\F\otimes\H})[\sigma'_{u,u}] \, (\sigma_{v,t}\otimes I_\H)^\dag(I_{\H\otimes\E_\F}\otimes\mathfrak S) \right),
\end{longtable}
 where ${\cal I}_{\E_\F\otimes\H}$ is the identity map over $\L(\E_\F\otimes\H)$, ${\mathfrak S}=\ket{\psi_1}\otimes \ket{\psi_2}\mapsto\ket{\psi_2}\otimes \ket{\psi_1}$ is the swap operator (here acting on $\H\otimes\H$), $\sigma_{u,u}=\input{HECurhou.tikz}$, $\sigma'_{u,u}=\input{HHEFHsigmauu.tikz}$, $\sigma_{v,t}=\input{HHEFvt.tikz}$, and $U^{(2)} = \input{U2_larger.tikz}$ as in Definition~\ref{def:S2T2}.

Again, similarly,
\begin{align*}
&\interp{C[U',\ket{\varepsilon'},\E']}(\ket{c,p}\!\bra{c'\!,p'}\otimes \rho) = \ket{c_{c,p}^D,p_{c,p}^D}\!\bra{c_{c'\!,p'}^D,p_{c'\!,p'}^D} \otimes \\[2mm]
& \textup{Tr}_{\mathcal E_\F,\H,\H} \left( (\sigma_{v,t}\otimes I_\H) \, (\S^{(2)}_{[U',\ket{\varepsilon'},\E']} \otimes {\cal I}_{\E_\F\otimes\H})[\sigma'_{u,u}] \, (\sigma_{v,t}\otimes I_\H)^\dag(I_{\H\otimes\E_\F}\otimes\mathfrak S_{\H,\H}) \right).
\end{align*}
Since $\S^{(2)}_{[U,\ket \varepsilon,\E]}=\S^{(2)}_{[U',\ket{\varepsilon'},\E']}$, this is equal to $\interp{C[U,\ket{\varepsilon},\E]}(\ket{c,p}\!\bra{c'\!,p'}\otimes \rho)$.

\end{itemize}

\paragraph*{Proof of necessity ($\neg\eqref{sameT1S2T2}\Rightarrow\neg\eqref{contextequiv1level2}$)}

\begin{itemize}

\item If $T^{(1)}_{[U,\ket{\varepsilon},\E]}\neq T^{(1)}_{[U',\ket{\varepsilon}',\E']}$, then by \cref{caracobsequivnegfree}, $[U,\ket{\varepsilon},\E]$ and $[U',\ket{\varepsilon}',\E']$ can be distinguished using a $\begin{tikzpicture}[scale=0.8]
	\begin{pgfonlayer}{nodelayer}
		\node [style=cartouche] (0) at (0, 0) {\tiny$\neg$};
		\node [style=none] (1) at (-0.8, 0) {};
		\node [style=none] (2) at (0.8, 0) {};
	\end{pgfonlayer}
	\begin{pgfonlayer}{edgelayer}
		\draw [style=noire] (1.center) to (0);
		\draw [style=noire] (0) to (2.center);
	\end{pgfonlayer}
\end{tikzpicture}
$-free context $C[\cdot]:\H^{(1)}$, so in particular, $[U,\ket{\varepsilon},\E]\not\approx_2[U',\ket{\varepsilon}',\E']$.

\item If $\S^{(2)}_{[U,\ket{\varepsilon},\E]}\neq \S^{(2)}_{[U',\ket{\varepsilon'},\E']}$, then one can distinguish $[U,\ket{\varepsilon},\E]$ and $[U',\ket{\varepsilon}',\E']$ as follows.

By assumption, there exists $\ket \varphi\in\H\otimes\H$ s.t. $\rho\neq \rho'$, where $\rho,\rho'\in\Lin(\H\otimes\H)$ are defined as follows: 
\[\rho=\ \input{SP2_L-ground-phi.tikz}\qquad\qquad\rho'=\ \input{SP2_R-ground-phi.tikz}\]
Let then $W_0$ be a unitary in $\L(\H^{\otimes 2})$ such that $W_0\ket{00} = \ket\varphi$. 

Note, on the other hand, that matrices of the form $W_1^\dagger(\ket{0}\!\bra{0}\otimes I_\H)W_1$, for all unitaries $W_1\in\L(\H^{\otimes 2})$, span the whole space $\L(\H^{\otimes 2})$.%
\footnote{This can be seen for instance explicitly by noting that any $\ket{\phi}\!\bra{\phi} = V_\phi \ket{00}\!\bra{00}V_\phi^\dagger$ (which themselves span the whole space, for some unitaries $V_\phi$) can be decomposed onto vectors of the form $W_1^\dagger(\ket{0}\!\bra{0}\otimes I_\H)W_1$, as $V_\phi \ket{00}\!\bra{00}V_\phi^\dagger = V_\phi [ (\ket{0}\!\bra{0}\otimes I_\H) - \frac{1}{d}\sum_{i=0}^{d-1} V_i (\ket{0}\!\bra{0}\otimes I_\H) V_i^\dagger + \frac{1}{d} \mathfrak{S} (\ket{0}\!\bra{0}\otimes I_\H) \mathfrak{S}^\dagger ] V_\phi^\dagger$, where $d$ is the dimension of $\H$, with the unitaries $V_i = I_{\H^{\otimes 2}} - (\ket{0}-\ket{i})(\bra{0}-\bra{i})\otimes \ket{0}\!\bra{0}$ (such that $V_i (\ket{0}\!\bra{0}\otimes I_\H) V_i^\dagger = \ket{0}\!\bra{0}\otimes I_\H + (\ket{i}\!\bra{i}-\ket{0}\!\bra{0})\otimes\ket{0}\!\bra{0}$) and the swap operator $\mathfrak{S}$.}
It follows that for $\rho \neq \rho'$ defined above, there exists a unitary $W_1$ such that $\input{rhoaccoladeW1terre.tikz} \neq \input{rhoprimeaccoladeW1terre.tikz}$.%
\footnote{Indeed: assume, by contradiction, that $\input{rhoaccoladeW1terre.tikz} = \input{rhoprimeaccoladeW1terre.tikz}$ for all unitaries $W_1$. Then in particular (by projecting the output wire onto $\ket{0}$) one has $\textup{Tr}[\rho W_1^\dagger(\ket{0}\!\bra{0}\otimes I_\H)W_1] = \textup{Tr}[\rho' W_1^\dagger(\ket{0}\!\bra{0}\otimes I_\H)W_1]$ for all $W_1$. Given, as just noted, that the matrices $W_1^\dagger(\ket{0}\!\bra{0}\otimes I_\H)W_1$ span the whole space $\L(\H^{\otimes 2})$, one concludes that $\rho=\rho'$---in contradiction with the fact that $\rho\neq\rho'$.}

\bigskip

In order to distinguish the two purified channels $[U,\ket\varepsilon,\E]$ and $[U',\ket{\varepsilon'},\E']$, we then consider the following context:
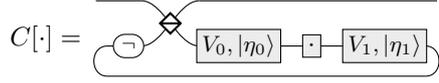
\[C[\cdot]=\input{contexteV0V1V0V1variante.tikz}\]
\noindent where $V_0 = \input{contrW0puisswapsurX-white.tikz}$, $V_1 =\input{swapsurXpuiscontrW1-white.tikz}$, with $W_0, W_1$ just introduced and $\ket{\eta_0}=\ket{\eta_1} = \ket0\otimes \ket0\in\H\otimes\CC^2$.\footnote{Where, given $W\in\Lin(\mathcal K)$, the controlled linear operation $\input{contrW-white.tikz}$ is defined as $W\otimes\ket0\!\bra0+I_{\mathcal K}\otimes\ket1\bra1$, and where $X=\begin{psmallmatrix}0&1\\[1mm]1&0\end{psmallmatrix}$.} \bigskip

One then has

$$\interp{C[U,\ket{\varepsilon},\E]}(\ket{\uparrow}\!\bra{\uparrow}\otimes \ket {0}\!\bra{0}) = \input{interp2.tikz}$$

\[=\ \input{interp2expanded-white_W0W1.tikz}\]

\[=\ \input{interp2expanded2_W0W1.tikz}\]

\[=\ \input{interp2expanded3_W0W1.tikz}\]

\[=\ \input{rhoaccoladeW1terre.tikz} \qquad \neq \qquad \input{rhoprimeaccoladeW1terre.tikz} \quad = \quad \interp{C[U',\ket{\varepsilon'},\E']}(\ket{\uparrow}\!\bra{\uparrow}\otimes \ket {0}\!\bra{0}).\]

Hence $\interp{C[U,\ket{\varepsilon},\E]} \neq \interp{C[U',\ket{\varepsilon'},\E']}$, and therefore $[U,\ket{\varepsilon},\E]\not\approx_2[U',\ket{\varepsilon}',\E']$. 

\item If $T^{(2)}_{[U,\ket{\varepsilon},\E]}\neq T^{(2)}_{[U',\ket{\varepsilon}',\E']}$, then let us first introduce the following lemma:
\begin{lemma}\label{lemma:T2v2}
Given two purified channels $[U,\ket{\varepsilon},\E]$ and $[U',\ket{\varepsilon}',\E']$, $T^{(2)}_{[U,\ket{\varepsilon},\E]}= T^{(2)}_{[U',\ket{\varepsilon}',\E']}$ if and only if for any $V\in\L(\H)$,
\[\input{TM2_L-V.tikz}\ =\ \input{TM2_R-V.tikz}.\]
\end{lemma}
\begin{proof}~
\begin{longtable}{CC@{}CCC}
&&\input{TM2_L.tikz}&=&\input{TM2_R.tikz}\\\\
\Leftrightarrow&\forall i,j,&\input{TM2_L-ij.tikz}&=&\input{TM2_R-ij.tikz}\\\\
\Leftrightarrow&\forall i,j,&\input{TM2_L-ijfilstires.tikz}&=&\input{TM2_R-ijfilstires.tikz}\\\\
\Leftrightarrow&\forall i,j,&\input{TM2_L-ketibraj.tikz}&=&\input{TM2_R-ketibraj.tikz}\\\\
\Leftrightarrow&\forall V\in\L(\H),&\input{TM2_L-V.tikz}&=&\input{TM2_R-V.tikz}\\\\
\end{longtable}
\end{proof}
By this lemma, since unitary operators span the whole space ${\mathcal L}(\H)$, if $T^{(2)}_{[U,\ket{\varepsilon},\E]}\neq T^{(2)}_{[U',\ket{\varepsilon}',\E']}$ then there exists a unitary operator $V\in\L(\H)$ such that
\[\input{TM2_L-V.tikz}\ \neq\ \input{TM2_R-V.tikz}.\]
Then by considering the following context: \[C[\cdot]=\input{boucletraversetrouV,1trou.tikz}\]
one gets \[\interp{C[U,\ket{\varepsilon},\E]}(\ket{\uparrow,0}\!\bra{\rightarrow,0}\otimes I_\H)=\ket{\uparrow,0}\!\bra{\rightarrow,0}\otimes\input{TM2_L-V.tikz}\] whereas \[\interp{C[U',\ket{\varepsilon'},\E']}(\ket{\uparrow,0}\!\bra{\rightarrow,0}\otimes I_\H)=\ket{\uparrow,0}\!\bra{\rightarrow,0}\otimes\input{TM2_R-V.tikz}.\] Hence $\interp{C[U,\ket{\varepsilon},\E]} \neq \interp{C[U',\ket{\varepsilon'},\E']}$, which proves that $[U,\ket{\varepsilon},\E]\not\approx_2[U',\ket{\varepsilon}',\E']$.

\end{itemize}

\end{proof}

\subsection{Proof of Proposition~\ref{prop:diff}}
\label{appendix:equiv}

\propdiff*

\begin{proof}
{[$\approx_{iso} ~\subseteq~ \approx_2$]} 
Since $ \approx_2$ is an equivalence relation it is enough to show that $\triangleleft_{iso}\subseteq \approx_2$. If $[U,\ket{\varepsilon},\E] \triangleleft_{iso} [U',\ket{\varepsilon'},\E']$, then the three conditions of Theorem \ref{caracobsequivlevel2} are satisfied, implying  $[U,\ket{\varepsilon},\E] \approx_2 [U',\ket{\varepsilon'},\E']$.
\\
{[$\approx_{2} ~\neq~ \approx_{iso}$]}
We consider the  following two purified $\mathbb C$-channels: $[X,\ket{0},\mathbb C^3]$ and $[XN,\ket{0},\mathbb C^3]$ where $X=\ket{x}\mapsto \ket{x{-}1\bmod 3}$ and $N = \ket x \mapsto (-1)^x \ket x$ are two (qutrit) unitary transformations. The two purified channels are $\approx_2$-equivalent as they satisfy the conditions of Theorem~\ref{caracobsequivlevel2}. In order to show that they are not iso-equivalent, note that if two purified $\mathbb C$-channels $[U,\ket{\varepsilon},\E]$ and $[U',\ket{\varepsilon'},\E']$ are iso-equivalent then for any $k\ge 0$ one has $\bra{\varepsilon}U^k\ket{\varepsilon} = \bra{\varepsilon'}WU^k\ket{\varepsilon}= \bra{\varepsilon'}U'^kW\ket{\varepsilon} = \bra{\varepsilon'}U'^k\ket{\varepsilon'}$. Since $\bra 0 X^3\ket 0 =1 \neq -1 = \bra 0 (XN)^3\ket 0$, it follows that $[X,\ket{0},\mathbb C]$ and $[XN,\ket{0},\mathbb C]$ are indeed not iso-equivalent.
\\
{[$\approx_2 ~\subsetneq~ \approx_1~\subsetneq~ \approx_0$]} The inclusions are clear from the characterisations of Theorems~\ref{thm:cptpchannels}, \ref{caracobsequivnegfree} and~\ref{caracobsequivlevel2}, together with Remark~\ref{SP2impliesSP1}. The fact that the inclusions are strict follows from the observations that the various conditions appearing in these theorems are non-redundant.
\end{proof}

\end{document}